\documentclass[prx,twocolumn,superscriptaddress,nofootinbib,showpacs,notitlepage]{revtex4-2}
\usepackage{amsmath}
\usepackage{amsfonts}
\usepackage{amsthm}
\usepackage{amssymb}
\usepackage{bbm}
\usepackage[utf8]{inputenc}
\usepackage{braket}
\usepackage{environ}
\usepackage{graphicx,color,dsfont}
\usepackage{booktabs}
\usepackage{floatrow}
\usepackage[caption=false,label font={bf,normalsize}]{subfig}
\usepackage[dvipsnames]{xcolor}
\usepackage[all]{xy}
\usepackage{relsize}
\usepackage{hyperref}
\usepackage{accents}
\hypersetup{linktocpage}
\definecolor{darkblue}{RGB}{0,0,127} 
\definecolor{darkgreen}{RGB}{0,150,0}
\hypersetup{colorlinks=true,citecolor=darkblue,linkcolor=darkblue, 
	filecolor=red, urlcolor=darkblue,pdftitle={Emergent fermionic gauge theory and foliated fracton order in the Chamon model},
pdfauthor={Wilbur Shirley, Xu Liu, Arpit Dua}}
\usepackage{lmodern}
\usepackage{amsmath}
\makeatletter
\renewcommand*\env@matrix[1][*\c@MaxMatrixCols c]{%
  \hskip -\arraycolsep
  \let\@ifnextchar\new@ifnextchar
  \array{#1}}
\makeatother

\newtheorem{theorem}{Theorem}[section]
\newtheorem{lemma}{Lemma}[section]

\usepackage{lipsum}

\makeatletter
\def\@opargbegintheorem#1#2#3{\trivlist
   \item[]{\bfseries #1\ #2\ (#3)} \itshape}
\makeatother

\NewEnviron{eqs}{%
\begin{equation}\begin{split}
    \BODY
\end{split}\end{equation}
}

\newcommand{\Chamon}{Chamon }

\newcommand{\beq}{\begin{eqnarray}}
\newcommand{\eeq}{\end{eqnarray}}
\newcommand{\FF}{{\mathbb{F}}}
\newcommand{\im}{{\mathrm{im}}}

\usepackage{hyperref}
\definecolor{darkblue}{RGB}{0,0,127} 
\definecolor{darkgreen}{RGB}{0,150,0}
\hypersetup{
	colorlinks, 
	linkcolor=darkblue, 
	citecolor=darkgreen, 
	filecolor=red, 
	urlcolor=blue,
  pdfauthor = {Wilbur Shirley, Xu Liu, Arpit Dua},
  pdftitle = {Chamon}
}

\definecolor{ultramarine}{RGB}{0,150,96}

\begin{document}

\title{Emergent fermionic gauge theory and foliated fracton order in the Chamon model}

\author{Wilbur Shirley}
\affiliation{School of Natural Sciences, Institute for Advanced Study, Princeton, NJ 08540, USA}
\affiliation{Department of Physics and Institute for Quantum Information and Matter, \mbox{California Institute of Technology, Pasadena, California 91125, USA}}
\author{Xu Liu}
\thanks{Current Address: Facebook Inc., Menlo Park, CA 94025, USA}
\affiliation{\mbox{Mani L. Bhaumik Institute for Theoretical Physics, Department of Physics and Astronomy}, \mbox{University of California at Los Angeles, Los Angeles, California 90095, USA}}

\author{Arpit Dua}
\affiliation{Department of Physics and Institute for Quantum Information and Matter, \mbox{California Institute of Technology, Pasadena, California 91125, USA}}

\begin{abstract}
    The \Chamon model is an exactly solvable spin Hamiltonian exhibiting nontrivial fracton order. In this work, we dissect two distinct aspects of the model. First, we show that it exhibits an emergent fractonic gauge theory coupled to a fermionic subsystem symmetry-protected topological state under four stacks of $\mathbb{Z}_2$ planar symmetries. Second, we show that the \Chamon model hosts 4-foliated fracton order by describing an entanglement renormalization group transformation that exfoliates four separate stacks of 2D toric codes from the bulk system.
    
\end{abstract}

\maketitle

\section{Introduction}
Gapped quantum systems can form nontrivial phases of matter in the absence of symmetry if they exhibit long-range entanglement in the many-body ground state~\cite{ChenLRE2010}. The traditional examples of long-range entangled phases are those with intrinsic topological order such as fractional quantum Hall states~\cite{Laughlin_1983,Wen_1995} and discrete gauge theories~\cite{dePropitius1999,qdouble}, which are characterized at low energy by topological quantum field theories~\cite{Witten_TQFT}. In 2005, Chamon discovered a three-dimensional exactly solvable lattice model~\cite{Chamon2005quantum} that represents the first example of a new kind of long-range entangled order known as \textit{fractonic order}~\cite{vijay_CB_2015,vijay2016fracton}. 

Quantum phases with fractonic order cannot be described by topological quantum field theory due to an intertwining of universal properties with lattice geometry~\cite{Gorantla2021,vijay_CB_2015,vijay2016fracton,PhysRevB.97.165106}. In particular, fractonic orders are characterized by a ground state degeneracy that scales exponentially with linear system size, and the existence of fractional excitations with constrained mobility~\cite{bravyi2011topological,Haah_code_2011,haah2013_thesis,haah2013commuting,vijay2016fracton}. The Chamon model, for instance, harbors three kinds of quasiparticles: \textit{planons}, which are mobile within a plane, \textit{lineons}, which can move along a line, and \textit{fractons}, which are fundamentally immobile as individual particles~\cite{bravyi2011topological}. In recent years, a wide range of fracton orders have been discovered theoretically, each exhibiting a different manifestation of constrained quasiparticle mobility and subextensive ground state degeneracy~\cite{Castelnovo_Chamon_2011,PhysRevB.81.184303,haah2011local,PhysRevLett.107.150504,kim20123d,yoshida2013exotic,bravyi2013quantum,haah2013_thesis,haah2013commuting,haah2014bifurcation, PhysRevLett.116.027202,vijay_CB_2015,vijay2016fracton,PhysRevB.95.245126,vijay2017isotropic,vijay2017generalization,PhysRevB.96.165106,PhysRevB.96.195139,HHB_models,hsieh_halasz_partons,PhysRevB.97.155111,PhysRevB.97.041110,PhysRevB.97.165106,Bulmash2018,PhysRevB.97.125102,PhysRevB.97.125101,PhysRevB.97.144106,finite_temp_Xcube,you2018symmetric,prem2018cage,Brown2019,Schmitz2019,You2019,You2019b,Weinstein2019,Prem2019,Bulmash2019,hao_twisted,Tantivasadakarn2,Tantivasadakarn_2021_nonabelian}. 
Notable examples include the Haah cubic code~\cite{haah2011local} and the X-cube model~\cite{vijay2016fracton}. It is natural to ask how the variety of fractonic orders can be systematically characterized within a common theoretical framework.

Many fractonic orders have a unified characterization as emergent gauge theories of discrete \textit{subsystem} symmetries, which have either planar or fractal geometry~\cite{vijay2016fracton,Williamson_cubic_code,shirley2018FoliatedFracton,kubica2018ungauging}. For example, the X-cube model is obtained by gauging three orthogonal sets of planar Ising symmetries of a cubic lattice spin-1/2 paramagnet (referred to as a \textit{3-foliated} gauge theory)~\cite{vijay2016fracton}.
The gauging procedure has been extended to fermion parity subsystem symmetries in fermionic systems, whose gauging yields gapped fractonic gauge theories with emergent fermionic charges~\cite{shirley2020fractonic,Tantivasadakarn1,}. On one hand, a large class of fractonic orders, including those belonging to the class of Calderbank-Shor-Steane (CSS) stabilizer codes, can be obtained via this procedure~\cite{kubica2018ungauging}. On the other hand, it remains unclear how, or if, certain fracton models including the \Chamon model, can be obtained by gauging and hence characterized by emergent gauge theory.

In a parallel development, the concept of \textit{foliated fracton order} (FFO) was recently introduced in an effort to systematically characterize fractonic orders with planon excitations~\cite{shirley2017fracton,shirley2018FoliatedFracton}. A lattice model is said to have FFO if the lattice size can be systematically reduced by removing, or \textit{exfoliating}, layers of 2D topological orders from the bulk 3D system via a finite-depth quantum circuit. Such a transformation maps a subset of the bulk planon excitations into anyons of the exfoliated 2D orders. For instance, for the X-cube model, it is possible to exfoliate layers of 2D toric code normal to the three cubic lattice directions, hence the X-cube model is said to have a \textit{3-foliation} structure. The notion of FFO has been shown to apply to a large class of models beyond the X-cube model~\cite{shirley2018CB,Wang2019,Shirley2019}. However, thus far it has remained unknown whether the fractonic order of the \Chamon model is foliated.

The purpose of this paper is to fill the gaps in the fracton literature by presenting two new results on the \Chamon model. First, we show that the model is characterized by a 4-foliated gauge theory coupled to a fermionic subsystem symmetry-protected topological (SSPT) state. In other words, it can be obtained by gauging four sets of planar $\mathbb{Z}_2$ symmetries that protect a non-trivial SSPT state~\cite{you2018subsystem,you2018symmetric} in a fermionic lattice system and then performing a local unitary transformation. This is a surprising result because there is no \textit{a priori} clear division of fractional excitations into gauge charge and gauge flux sectors (as is the case for CSS codes). Instead, it is necessary to expand the unit cell and divide the excitations into charge and flux sectors according to the sublattice on which they reside. This is reminiscent of the gauge theory description of the much simpler 2D Wen plaquette model~\cite{Wen2003}.

Second, we show that the Chamon model exhibits FFO with a 4-foliation structure composed of 2D toric code resource layers. In particular, we describe an entanglement renormalization group transformation~\cite{Vidal2007,haah2014bifurcation,Dua_RG_2019} that maps a copy of the \Chamon model on a $3L\times 3L\times 3L$ cubic lattice to a coarse-grained \Chamon model on an $L\times L\times L$ lattice tensored with four decoupled stacks of 2D toric codes. This 4-foliation structure is consistent with the four orientations of planons in the \Chamon model, and is most easily described in terms of the action of the transformation on the planon excitations. We have also obtained an explicit translation-invariant Clifford circuit realizing this transformation.


The paper is organized as follows. In Sec. \ref{sec:Chamon}, we review the \Chamon model and its essential properties. In Sec. \ref{sec:GaugeTheory}, we explain the characterization of the \Chamon model in terms of emergent fermionic gauge theory. In Sec. \ref{sec:FFO}, we describe the FFO exhibited by the \Chamon model. We conclude with a discussion in Sec. \ref{sec:Discussion}.

\section{The \Chamon model}
\label{sec:Chamon}

\begin{figure*}[tbp]
    \centering
    \includegraphics[width=\textwidth]{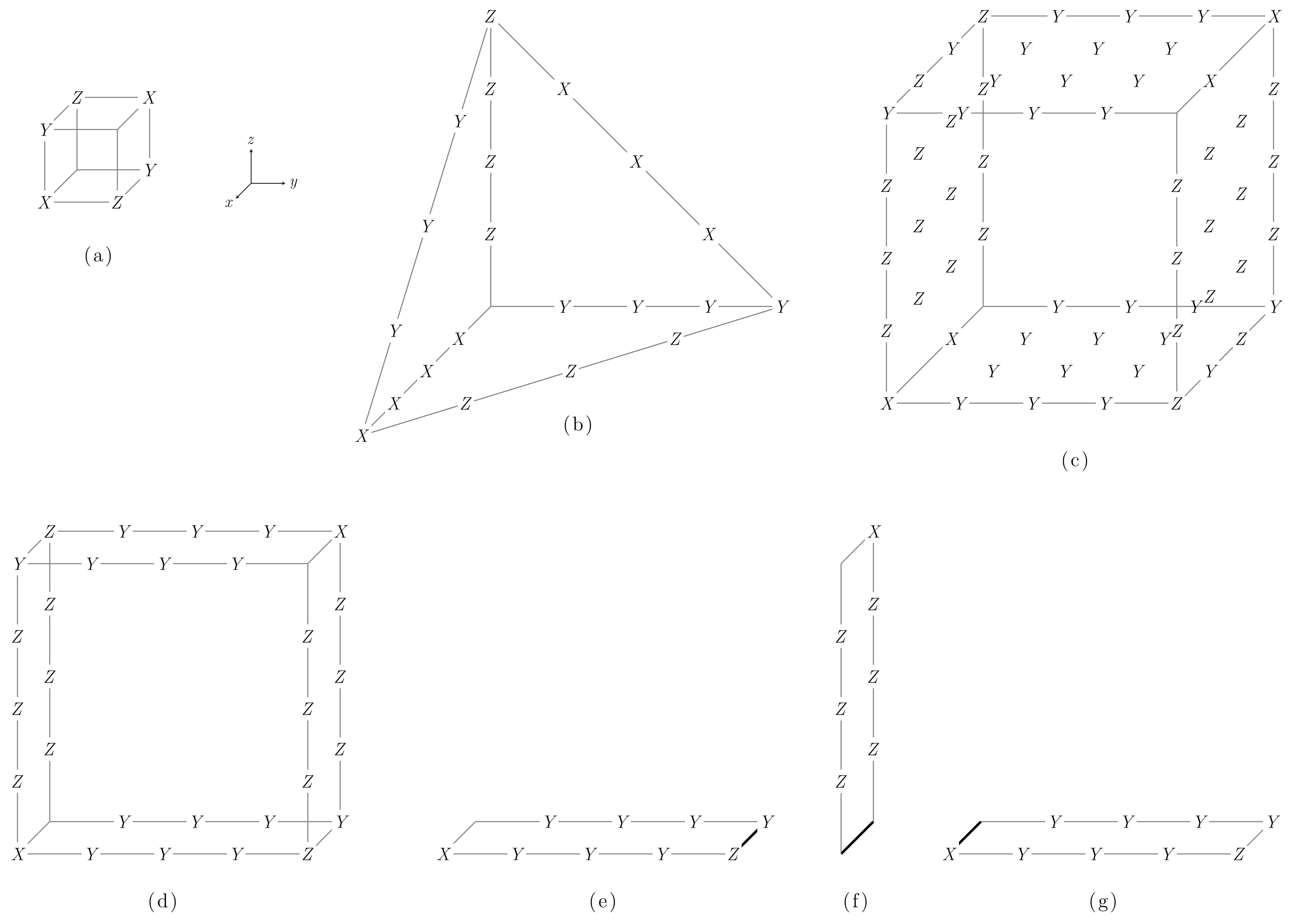}
    \caption{(a) The operator $O_c$, which is a tensor product of the six Pauli operators. (b) A tetrahedral wireframe operator, which is equal to a product of $O_c$ operators inside the tetrahedron. (c) A stabilizer operator bounding a cubic region. (d) A loop operator for an elementary planon of the Chamon model, which is a product of $O_c$ operators within the loop. (e-g) Three planon string operators $W_3$, $W_2$, and $W_1$ forming a T-junction (the bold edge represents the same edge in each subfigure). The fermionic exchange statistic of the elementary planon is given by $W_3W_2^\dagger W_1W_3^\dagger W_2W_1^\dagger =-1$.}
    \label{fig:ChamonOps}
\end{figure*}

The \Chamon model was originally defined on an FCC lattice with one qubit per site~\cite{Chamon2005quantum,bravyi2011topological}, exhibiting the tetrahedral point group symmetry of the lattice. For our purposes it will be more convenient to place the model on a cubic lattice with one qubit per site, by performing an isometry of $\mathbb{R}_3$ defined by
\begin{equation}\begin{split}
    \textstyle\left(0,\frac{1}{2},\frac{1}{2}\right)&\to(1,0,0)\\\textstyle\left(\frac{1}{2},0,\frac{1}{2}\right)&\to(0,1,0)\\\textstyle\left(\frac{1}{2},\frac{1}{2},0\right)&\to(0,0,1).
\end{split}\end{equation}
In this formulation the Hamiltonian has the form
\begin{equation}
    H_C=-\sum_c O_c
\end{equation}
where $c$ indexes the elementary cubes of the lattices and $O_c$ is the six-body Pauli operator depicted in Fig.~\ref{fig:ChamonOps}(a). For any pair of cubes $c,c'$, it holds that $\left[O_c,O_{c'}
\right]=0$, thus $H_C$ is an exactly solvable stabilizer code Hamiltonian~\cite{gottesman1997stabilizer}.
The ground state degeneracy (GSD) of the model on an $L_x\times L_y\times L_z$ periodic cubic lattice has the form
\begin{equation}
    \log_2\text{GSD}=L_x+L_y+L_z+\gcd(L_x,L_y,L_z)-3,
    \label{eq:GSD1}
\end{equation}
The linear component of this formula arises from the following relations between stabilizer generators:
\begin{equation}
    \prod_{c\in P}O_c=1
    \label{eq:relations}
\end{equation}
where $P$ is any dual lattice plane normal to the $x$, $y$ $z$, or $w=(1,1,1)$ directions. This gives a total of $L_x+L_y+L_z+L_w$ relations, where $L_w=\gcd(L_x,L_y,L_z)$ is the number of planes normal to $w$ under periodic boundary conditions. It is straightforward to confirm the constant correction numerically.

The model hosts fractional excitations of all mobility types: fractons, lineons, and planons. The excitation structure can be understood by examining the form of stabilizer operators corresponding to processes of quasiparticle creation, movement, and re-annihilation. These operators are given by certain products of $O_c$ terms, of which there are two types. The first type of operator is a wireframe operator, which is a product of all $O_c$ terms within a polyhedral region bounded by $x$, $y$, $z$, and $w$ planes. Due to the relations in (\ref{eq:relations}), such operators are supported on the edges of the polyhedron, for instance the tetrahedral wireframe operator pictured in Fig.~\ref{fig:ChamonOps}(b). Lineon excitations are created at the endpoints of truncated wireframe operators. There are six kinds of lineons, with mobility in the $x$, $y$, $z$, $(0,1,-1,)$, $(-1,0,1)$, and $(1,-1,0)$ directions, respectively. The lineons obey triple fusion rules in which three distinctly oriented lineons fuse together into the vacuum, which is possible when their respective string operators form the corner of a wireframe operator. For example, both the $x$, $y$, $z$ lineon triple and the $(1,-1,0)$, $(0,1,-1)$, $z$ lineon triple fuse into the vacuum, whereas $x$, $y$, $(1,-1,0)$ and $x$, $y$, $(-1,0,1)$ triples do not.

The second kind of operator is slightly more complicated: it is given by the product of all $O_c$ terms lying in \textit{even} dual lattice planes normal to a particular direction $\mu=x,y,z,w$ within a polyhedral region bounded by $x$, $y$, $z$, and $w$ planes. Such operators are supported on all surfaces of the polyhedron except those normal to $\mu$. An example of such an operator bounding a cubic region is shown in Fig.~\ref{fig:ChamonOps}(c). Truncating to a single face gives a membrane operator which creates fracton excitations at its corners.

There are four types of planons mobile within planes normal to the $x$, $y$, $z$, and $w$ directions respectively. For each direction, there is one independent species of planon per lattice spacing, referred to as an elementary planon. The loop operator for an elementary planon can be obtained by taking the product of all $O_c$ operators in a large region within a single $x$, $y$, $z$, or $w$ plane, for instance as depicted in Fig.~\ref{fig:ChamonOps}(d). These elementary planons can be viewed as either lineon dipoles or as fracton dipoles, \textit{i.e.} composite excitations of a pair of adjacent lineons or fractons, since the loop operator can be viewed as either the first or second kind of operator described in the previous paragraph. Hence, any lineon or fracton can be regarded as a semi-infinite stack of planons, so the exchange and braiding statistics of fractional excitations in the model are entirely characterized by the planon statistics. There are two important features: first, each of the elementary planons have fermionic exchange statistics. Second, adjacent parallel planons have a mutual $\pi$ braiding statistic. These facts can be verified by examining the structure of the planon string operators as shown in Fig.~\ref{fig:ChamonOps}(e-g). Since the elementary planons can be regarded as lineon dipoles, this also implies that intersecting lineons have a mutual $\pi$ braiding statistic.

\section{Emergent fermionic gauge theory}
\label{sec:GaugeTheory}

In this section we demonstrate that the Chamon model is equivalent under a generalized local unitary transformation~\cite{ChenLRE2010} to a fractonic gauge theory coupled to a fermionic subsystem symmetry-protected topological (SSPT) state~\cite{you2018subsystem,you2018symmetric}. We begin with the SSPT matter Hamiltonian $H_M$, which is symmetric under four stacks of $\mathbb{Z}_2$ planar symmetries. We then gauge the symmetry to obtain a spin model $H_G$. Finally, we transform $H_G$ into the Chamon model $H_C$ via a generalized local unitary.

We also sketch an argument that $H_M$ is a weak SSPT in the sense of Refs.~\cite{subsystemphaserel,Devakul_2020}.

\subsection{Matter Hamiltonian}

First we describe the matter Hamiltonian $H_M$. We consider a cellulation of $\mathbb{R}_3$ obtained by slicing along lattice planes of integer spacing normal to the $x$, $y$, $z$, and $w=(1,1,1)$ directions. The $x$, $y$, and $z$ planes divide $\mathbb{R}_3$ into unit volume elementary cubes, and each cube is further sliced into three 3-cells by the $w$ planes: two types of tetrahedra and one octahedron, as pictured in Fig. \ref{fig:TetraOcta}. The Hilbert space of $H_M$ is composed of one fermionic orbital per tetrahedron and one qubit per octahedron. The Hamiltonian has the form
\begin{align}
    H_M&=-\sum_ti\gamma_t\gamma'_t-\sum_o\mathcal{X}_o
\end{align}
where $t$ indexes tetrahedra and $o$ octahedra, and
\begin{equation}
    \mathcal{X}_o\equiv X_o\prod_{a=0}^1\prod_{b=0}^1\prod_{c=0}^1\prod_{d=0}^1 Z_{o+a\hat{y}-b\hat{z}+c(1,-1,0)+d(-1,0,1)}
    \label{eq:SSPT}
\end{equation}
where $o+\Vec{r}$ represents the octahedron displaced from $o$ by $\Vec{r}$ (see Fig.~\ref{fig:SSPT}(a)). The terms of $H_M$ mutually commute, hence the model is exactly solvable.

\begin{figure}[tbp]
    \centering
    \includegraphics[width=.4\textwidth]{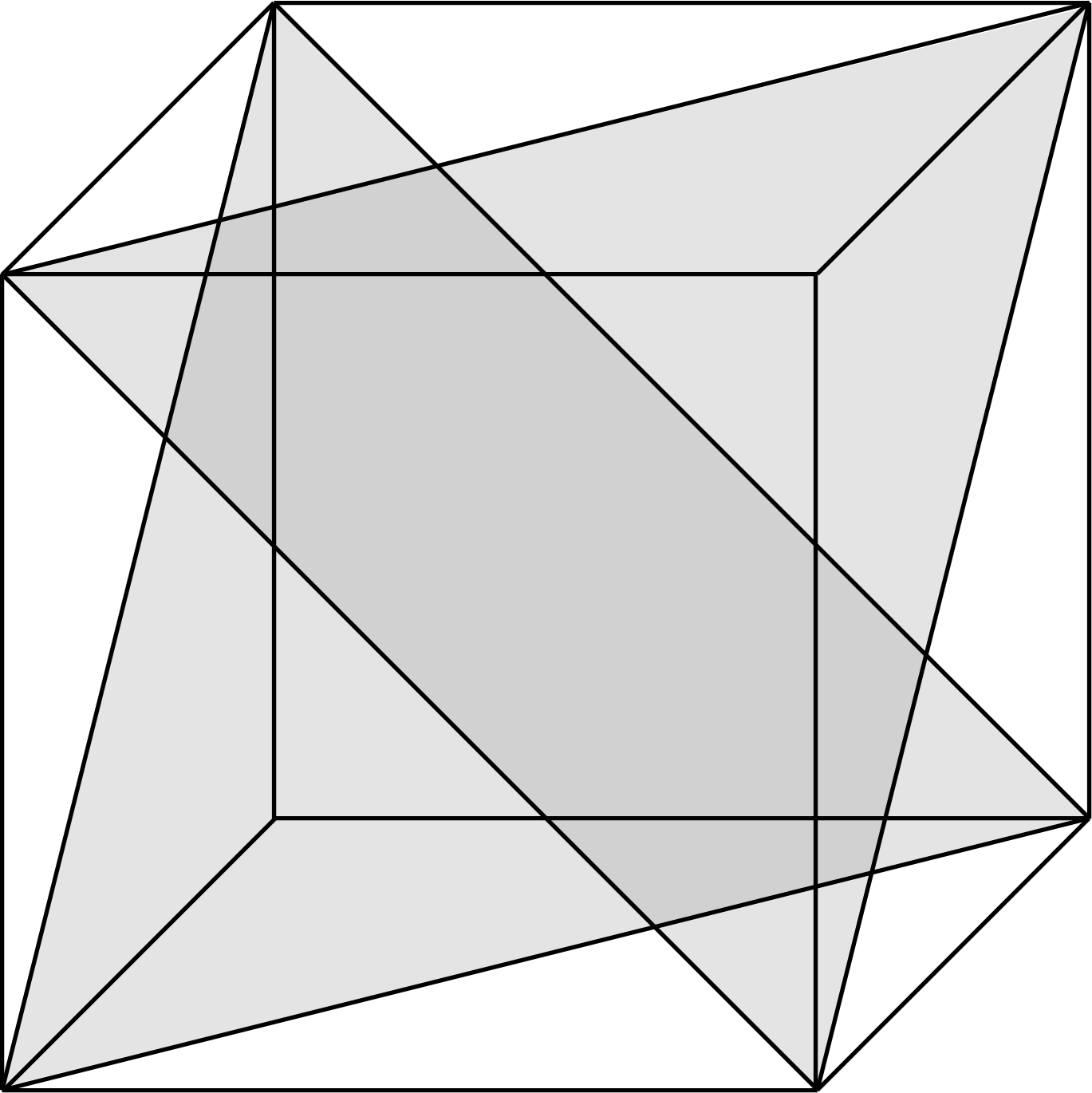}
    \caption{The tetrahedral-octahedral honeycomb. Each cube of a cubic lattice is split into two tetrahedra and one octahedron by $(1,1,1)$ planes (shaded).}
    \label{fig:TetraOcta}
\end{figure}

$H_M$ is symmetric under four stacks of unitary $\mathbb{Z}_2$ planar subsystem symmetries, normal to the $x$, $y$, $z$, and $w$ directions. Each symmetry generator is associated with a \textit{dual} lattice plane of the tetrahedral-octahedral honeycomb. Let $P$ denote the set of all 3-cells lying in a dual lattice plane. Then the corresponding symmetry of $H_M$ is
\begin{equation}
    S_P=\prod_{t\in P}i\gamma_t\gamma'_t\prod_{o\in P}X_o.
\end{equation}
There is one symmetry generator for every such $P$. To see that the $\mathcal{X}_o$ terms commute with all of these symmetries, note that each of the $x$, $y$, $z$, and $w$ planes contains at least one of the $\hat{y}$, $\hat{z}$, $(1,-1,0)$, or $(-1,0,1)$ vectors.

We note that the subsystem symmetries obey the global relations
\begin{equation}
    \prod_{P_x}{S}_{P_x}=\prod_{P_y}{S}_{P_y}=\prod_{P_z}{S}_{P_z}=\prod_{P_w}{S}_{P_w}
\end{equation}
where the products are over all dual lattice planes $P_\mu$ normal to $\mu$. Importantly, we also note that the product of symmetries over all \textit{even} dual lattice planes in all four directions is equal to the \textit{global} fermion parity $\mathbb{Z}_2^F$, which is thus generated by the subsystem symmetry group. Therefore, a bosonic system will be obtained upon gauging the symmetries.

\subsection{Gauging}
\label{sec:Gauging}

We now discuss the gauging of symmetries according to the general prescription~\cite{Levin_2012,shirley2018FoliatedFracton,shirley2020fractonic,Tantivasadakarn1}. The first step is to identify a set of `minimal couplings' that generate the algebra of symmetric operators together with the on-site symmetry representations (Pauli $X$ on qubits and $i\gamma\gamma'$ on fermion orbitals). There is one minimal coupling for each edge $e$ of the tetrahedral-octahedral honeycomb, acting on the degrees of freedom associated with the four 3-cells adjacent to $e$ (two octahedra $o$ and $o'$ and two oppositely oriented tetrahedra $t$ and $t'$), which we choose to be
\begin{equation}
    M_e\equiv Z_{o}Z_{o'}\gamma_{t}\gamma_{t'}.
    \label{eq:couplingfermionic}
\end{equation}

\begin{figure}[tbp]
    \centering
    \includegraphics[width=.95\textwidth]{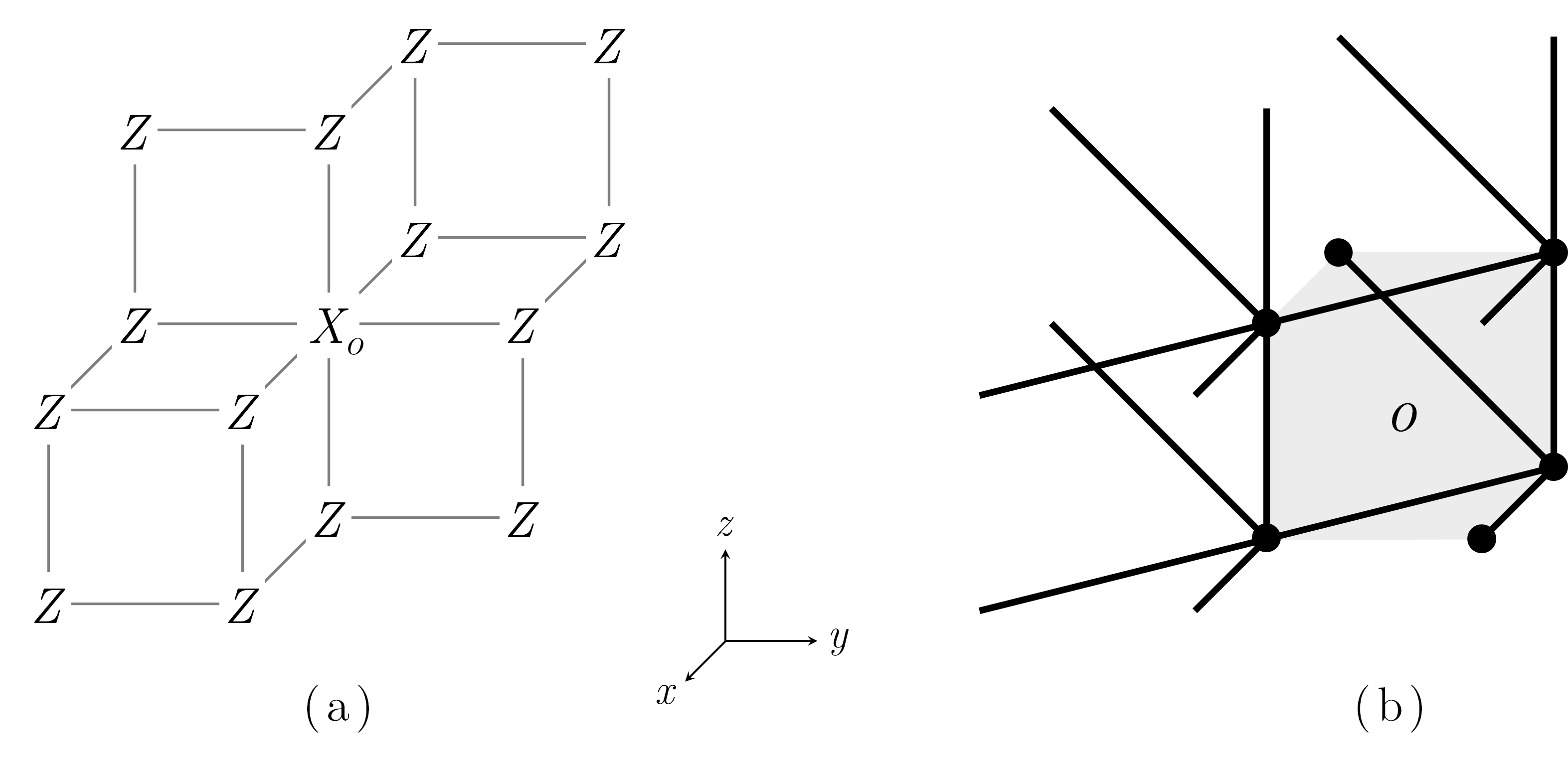}
    \caption{(a) Depiction of the operator $\mathcal{X}_o$. Each Pauli operator acts on an octahedral qubit, whose center-points form a cubic lattice. The octahedron $o$ is indicated by subscript. (b) The set of edges $E_o$ with respect to the octahedron $o$, whose vertices are the six dots.}
    \label{fig:SSPT}
\end{figure}

The second step is to introduce a gauge qubit degree of freedom for each minimal coupling, hence one per edge. We simultaneously restrict the Hilbert space by introducing generalized Gauss's law constraints for each matter degree of freedom. The constraints have the form
\begin{equation}
    X_o\prod_{e\in o}X_e=1,\qquad
    i\gamma_t\gamma'_t\prod_{e\in t}X_e=1
    \label{eq:constraintfermionic}
\end{equation}
for each octahedron $o$ and tetrahedron $t$.

The third step is to couple the gauge and matter degrees of freedom by introducing a gauged Hamiltonian that preserves the constraints. In particular, in the gauged Hamiltonian, the minimal coupling for each edge $e$ is composed with the gauge qubit operator $Z_e$:
\begin{equation}
    M_e\to M_eZ_e.
\end{equation}
This modification is non-unique, since there are multiple ways to express the operator $\mathcal{X}_o$ in terms of the minimal couplings. We choose the expression
\begin{equation}
    \mathcal{X}_o=X_o\prod_{e\in E_o}M_e
\end{equation}
where $E_o$ is the set of edges depicted in Fig.~\ref{fig:SSPT}(b). Hence
\begin{equation}
    \mathcal{X}_o\to X_o\prod_{e\in E_o}M_eZ_e.
\end{equation}

The final step is to add a set of terms $B_{v,\mu}$ for each vertex $v$ to the gauged Hamiltonian in order to gap out the gauge flux excitations. Here $\mu=x,y,z,w$ and $B_{v,\mu}$ is defined as the tensor product of Pauli $Z$ operators over the six links adjacent to $v$ in the plane normal to $\mu$. Thus, the gauged Hamiltonian takes the form
\begin{equation}
    \tilde{H}_M=-\sum_ti\gamma_t\gamma'_t-\sum_oX_o\prod_{e\in E_o}M_eZ_e-\sum_{v,\mu}B_{v,\mu},
\end{equation}
subject to the constraints (\ref{eq:constraintfermionic}).

The matter degrees of freedom can be eliminated via the unitary
\begin{equation}\begin{split}
    X_o\to X_o\prod_{e\in o}X_e\qquad Z_o\to Z_o\\
    \gamma_t\to\gamma_t\prod_{e\in t}X_e\qquad\gamma_t'\to\gamma_t'\\
    X_e\to X_e\qquad Z_e\to M_e\overline{Z}_e,
    \label{eq:unitary}
\end{split}\end{equation}
which maps the constraints of (\ref{eq:constraintfermionic}) to $X_o=1$ and $i\gamma_t\gamma_t'=1$ respectively. The $\overline{Z}_e$ operators are defined in Fig.~\ref{fig:big} such that $\overline{Z}_e$ and $\overline{Z}_{e'}$ anticommute if $e$ and $e'$ belong to the same tetrahedron, and commute otherwise. In the constrained space, $\tilde{H}_M$ is mapped to a bosonic Hamiltonian $H_G$ acting on the pure gauge qubit Hilbert space:
\begin{equation}
    H_G=-\sum_c\overline{A}_c-\sum_{v,\mu}\overline{B}_{v,\mu},
\end{equation}
where
\begin{equation}
    \overline{A}_t\equiv\prod_{e\in c}X_e,\qquad
    \overline{A}_o\equiv\prod_{e\in o}X_e\prod_{e\in E_o}\overline{Z}_e,
\end{equation}
and $\overline{B}_{v,\mu}$ is the image of $B_{v,\mu}$ under the unitary (\ref{eq:unitary}). The terms of $H_G$ mutually commute, hence they define a Pauli stabilizer code.

\begin{figure*}[tbp]
    \centering
    \includegraphics[width=\textwidth]{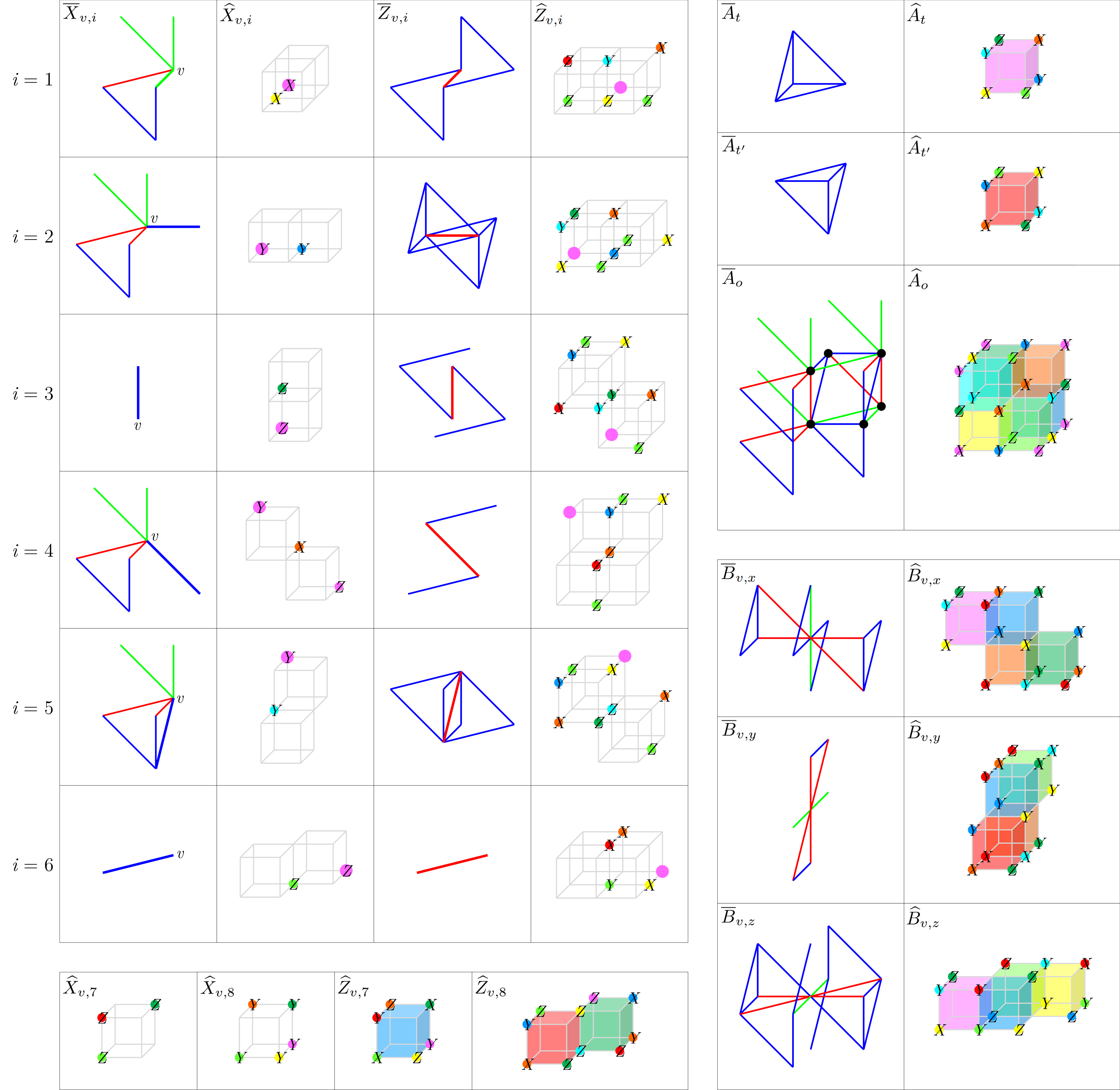}
    \caption{Definitions of the Pauli operators introduced in this section. The operators $\overline{X}_{v,i}$ and $\overline{Z}_{v,i}$ acting on $\mathcal{H}_G$ for $i=1,\ldots,6$ are defined in the table on the left, which are equivalent to the $\overline{X}_e$ and $\overline{Z}_e$ operators for the bold edge $e$. Red, green, and blue edges respectively represent the action of Pauli $Z$, $Y$, and $X$. The operators $\widehat{X}_{v,i}$ and $\widehat{Z}_{v,i}$ acting on $\mathcal{H}_C$ for $i=1,\ldots,8$ are defined in the tables on the left, with $v$ given by the enlarged magenta dot in each figure (an unlabelled enlarged dot has no Pauli action). The 3-cell operators $\overline{A}_t$, $\overline{A}_{t'}$, and $\overline{A}_o$, and vertex operators $\overline{B}_{v,\mu}$ of $H_G$ are defined in the tables on the right. The vertices of octahedron $o$ are indicated by black dots, whereas the vertex $v$ for each $\overline{B}_{v,\mu}$ operator is the central vertex. The operators $\widehat{A}_t$, $\widehat{A}_{t'}$, $\widehat{A}_o$, $\widehat{B}_{v,x}$, $\widehat{B}_{v,y}$, and $\widehat{B}_{v,z}$ acting on $\mathcal{H}_C$ are likewise defined in the tables on the right. These operators, together with $\widehat{Z}_{v,7}$ and $\widehat{Z}_{v,8}$, generate the stabilizer group of $H_C$. The shaded cubes indicate that a given operators is equal to a product of the corresponding cube terms of $H_C$ (the color of each cube corresponds to the vertex of minimum $x$, $y$, and $z$ coordinates).}
    \label{fig:big}
\end{figure*}

\subsection{Excitation content and ground state degeneracy of the gauged Hamiltonian}

To analyze the properties of $H_G$, it is helpful to express the Hamiltonian in terms of operators $\overline{X}_e$ and $\overline{Z}_e$ associated with edge $e$ of the tetrahedral-octahedral honeycomb. These operators are defined in Fig.\ref{fig:big}. We have already used the $\overline{Z}_e$ operators in the unitary (\ref{eq:unitary}). In particular,
\begin{equation}
    \overline{A}_c=\prod_{e\in c}\overline{X}_e,\qquad \overline{B}_{v,\mu}=\prod_{v\ni e\perp\mu}\overline{Z}_e
    \label{eq:overline}
\end{equation}
where the second product is over the six edges $e$ adjacent to $v$ in the plane normal to $\mu$. These operators are defined in Fig.~\ref{fig:big} and satisfy the relations
\begin{equation}
    \overline{X}_e^2=\overline{Z}_e^2=1,\qquad
    \Big\{\overline{X}_e,\overline{Z}_e\Big\}=\Big[\overline{X}_e,\overline{Z}_{e'}\Big]=0
    \label{eq:relationsHG}
\end{equation}
where $e$ and $e'$ are distinct edges. On the other hand, if $e$ and $e'$ are nearby, then it is generically the case that
\begin{equation}
    \Big[\overline{X}_e,\overline{X}_{e'}\Big]\neq0,\qquad\Big[\overline{Z}_e,\overline{Z}_{e'}\Big]\neq0.
    \label{eq:relationsHG2}
\end{equation}
It is instructive to note that due to (\ref{eq:overline}), there is a formal relation between $H_G$ and a certain 4-foliated version of the X-cube model, $H_{4XC}$, described in Appendix~\ref{app:4XC}. Roughly speaking, $H_G$ is obtained from $H_{4XC}$ by replacing $X_e\to\overline{X}_e$ and $Z_e\to\overline{Z}_e$.

$H_G$ has six qubits and six stabilizer generators per unit cell (since one of the four $\overline{B}_{v,\mu}$ terms is redundant). The stabilizer generators obey the following relations:
\begin{equation}
    \prod_{c\in P}\overline{A}_c=1,\qquad\prod_{v\in P'}\overline{B}_{v,\mu}=1,
\end{equation}
where $c\in P$ indexes all 3-cells in a dual lattice plane $P$, and $v\in P'$ indexes all vertices belonging to a direct lattice plane $P$. However, three of these relations are redundant, hence the ground state degeneracy (GSD) of $H_G$ on an $L_x\times L_y\times L_z$ lattice with periodic boundary conditions satisfies
\begin{equation}
    \log_2\text{GSD}=2L_x+2L_y+2L_z+2\gcd(L_x,L_y,L_z)-3.
    \label{eq:GSD2}
\end{equation}

\begin{figure*}[tbp]
    \centering
    \includegraphics[width=.7\textwidth]{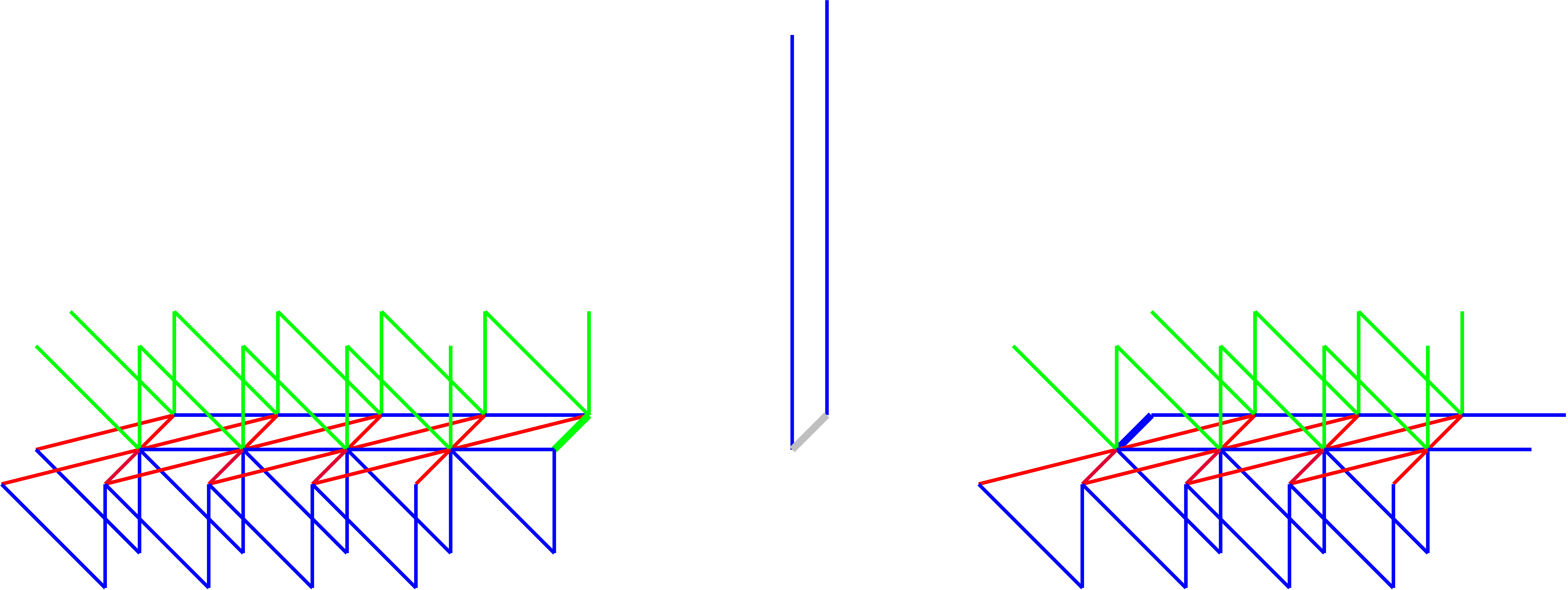}
    \caption{Three fracton dipole string operators $W_3$, $W_2$, and $W_1$ forming a T-junction (the bold edge represents the same edge in each subfigure). The fermionic exchange statistic of the fracton dipole is given by $W_3W_2^\dagger W_1W_3^\dagger W_2W_1^\dagger =-1$. The red, green and blue edges represent Pauli operators $Z$, $Y$ and $X$. The bold gray edge has no Pauli action.}
    \label{fig:exchangeHG}
\end{figure*}

The fractional excitations of $H_G$ can be split into two sectors, which we refer to as electric charges and magnetic fluxes.
The magnetic sector consists of lineons created at the endpoints of rigid string operators, which are finite segments of \textit{wireframe} operators equal to the product of all $\overline{A}_c$ terms within a polyhedral region bounded by $x$, $y$, $z$, and $w$ planes. Rigid string operators are equal to the product of $\overline{X}_e$ operators over all edges of the string, which follows from the first expression of (\ref{eq:overline}). There are six species of lineons, corresponding to the six orientations of edges in the tetrahedral-octahedral honeycomb: $x$, $y$, $z$, $(1,-1,0)$, $(0,1,-1)$, and $(-1,0,1)$. Triples of lineons meeting at a single vertex fuse into the vacuum if their string operators belong to the corner of a wireframe operator. For example, $x$, $y$, $z$, and $(1,-1,0)$, $(0,1,-1)$, $z$ lineon triples fuse into the vacuum, whereas $x$, $y$, $(1,-1,0)$ and $x$, $y$, $(-1,0,1)$ triples do not. Due to these triple fusion rules, composite excitations of two adjacent parallel lineons, i.e. lineon dipoles, are planons. There are four species of lineon dipoles in the model: those mobile in planes normal to the $x$, $y$, $z$, or $w$ directions. The loop operators for lineon dipoles are wireframe operators with a slab geometry.

The electric sector consists of fractons created at the corners of dual lattice membrane operators composed of a product of $\overline{Z}_e$ operators over all dual lattice faces comprising the membrane (each dual lattice face corresponds to a direct lattice edge $e$). Each fracton excitation is associated with a 3-cell of the tetrahedral-octahedral honeycomb. Fracton dipoles composed of a tetrahedral fracton and an adjacent octahedral fracton, are planons. There are four species of fracton dipoles in the model: those mobile in planes normal to the $x$, $y$, $z$, or $w$ directions.

The charge and flux sectors of $H_G$ interact via generalized long-range Aharanov-Bohm statistical interactions. In particular, a phase of $-1$ is obtained when a lineon dipole flux encircles a fractonic charge, and likewise when a fracton dipole charge encircles a lineonic flux. These interactions arise from the commutation relations of (\ref{eq:relationsHG}).

There are also nontrivial statistical interactions within both the electric and magnetic sectors, due to the nontrivial commutation relations of (\ref{eq:relationsHG2}). In the electric sector, the tetrahedral fractons are fermionic, whereas the octahedral fractons are bosonic. Therefore, each of the fracton dipoles is a fermion. This self-exchange statistic can be explicitly computed using the formula $\theta=W_3W_2^\dagger W_1W_3^\dagger W_2W_1^\dagger$ where $W_i$ are three fracton dipole string operators with a common endpoint~\cite{kitaev2006anyons,Levin_wen_fermion}, as in Fig.~\ref{fig:exchangeHG}.

\begin{figure}[tbp]
    \centering
    \includegraphics[width=.9\textwidth]{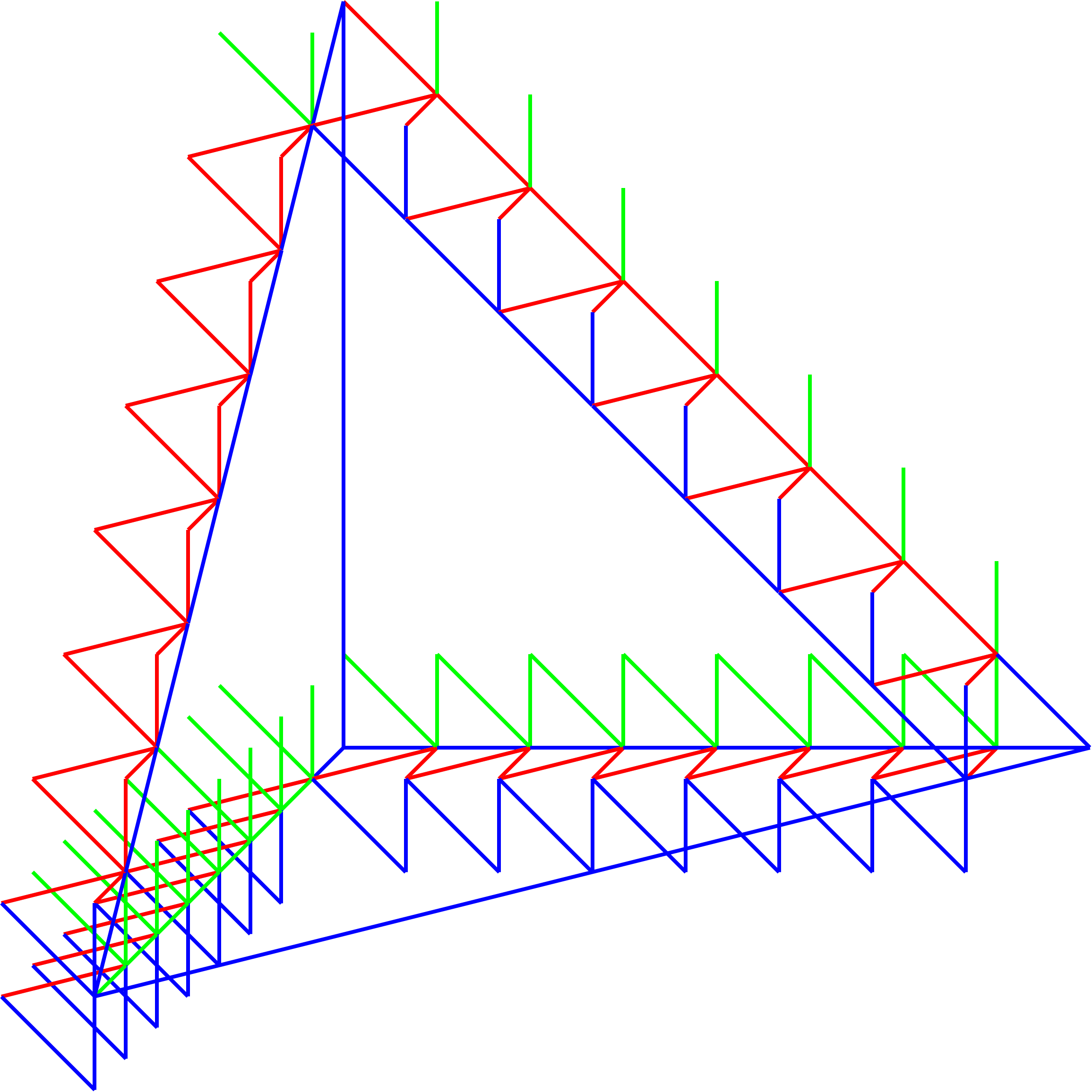}
    \caption{A tetrahedral wireframe operator for $H_G$, given by a product of $\overline{A}_c$ terms over 3-cells inside the tetrahedron. The red, green and blue edges represent Pauli operators $Z$, $Y$ and $X$.}
    \label{fig:wireframeHG}
\end{figure}

In the magnetic sector, the lineons exhibit nontrivial exchange statistics and nontrivial braiding statistics with other lineons. In particular, any pair of lineons intersecting in an $x$, $y$, $z$, or $w$ plane has a mutual $\pi$ braiding statistic, arising from the anticommutation of intersecting lineon string operators. This can be observed from the form of the wireframe operators, an example of which is shown in Fig.~\ref{fig:wireframeHG}. As a result, lineon dipoles in adjacent planes likewise have a $\pi$ braiding statistic. Moreover, each lineon dipole is a fermion.

\subsection{Mapping to the \Chamon model}
\label{sec:gLU}

We now describe a generalized local unitary (gLU) transformation that maps the ground space of $H_G$ to that of the \Chamon model $H_C$. Based on the expressions (\ref{eq:GSD1}) and (\ref{eq:GSD2}) for the ground state degeneracy of these models, it is clear that for this transformation to work, a unit cell of $H_G$ must correspond to a $2\times2\times2$ cell of $H_C$. Therefore, in this section, we place the \Chamon model qubits on the sites of a cubic lattice with \textit{half-integer} coordinates.
With respect to the integer cubic lattice, the \Chamon model has eight qubits and eight stabilizer generators per unit cell, forming a Hilbert space $\mathcal{H}_C$ as shown in Fig.~\ref{fig:ChamonUnitCell}. We label the qubits with a double subscript $v,i$ with $i=1,\ldots,8$ and $v$ the vertex of the integer lattice coinciding with qubit $1$.

\begin{figure}[tbp]
    \centering
    \includegraphics[width=.5\textwidth]{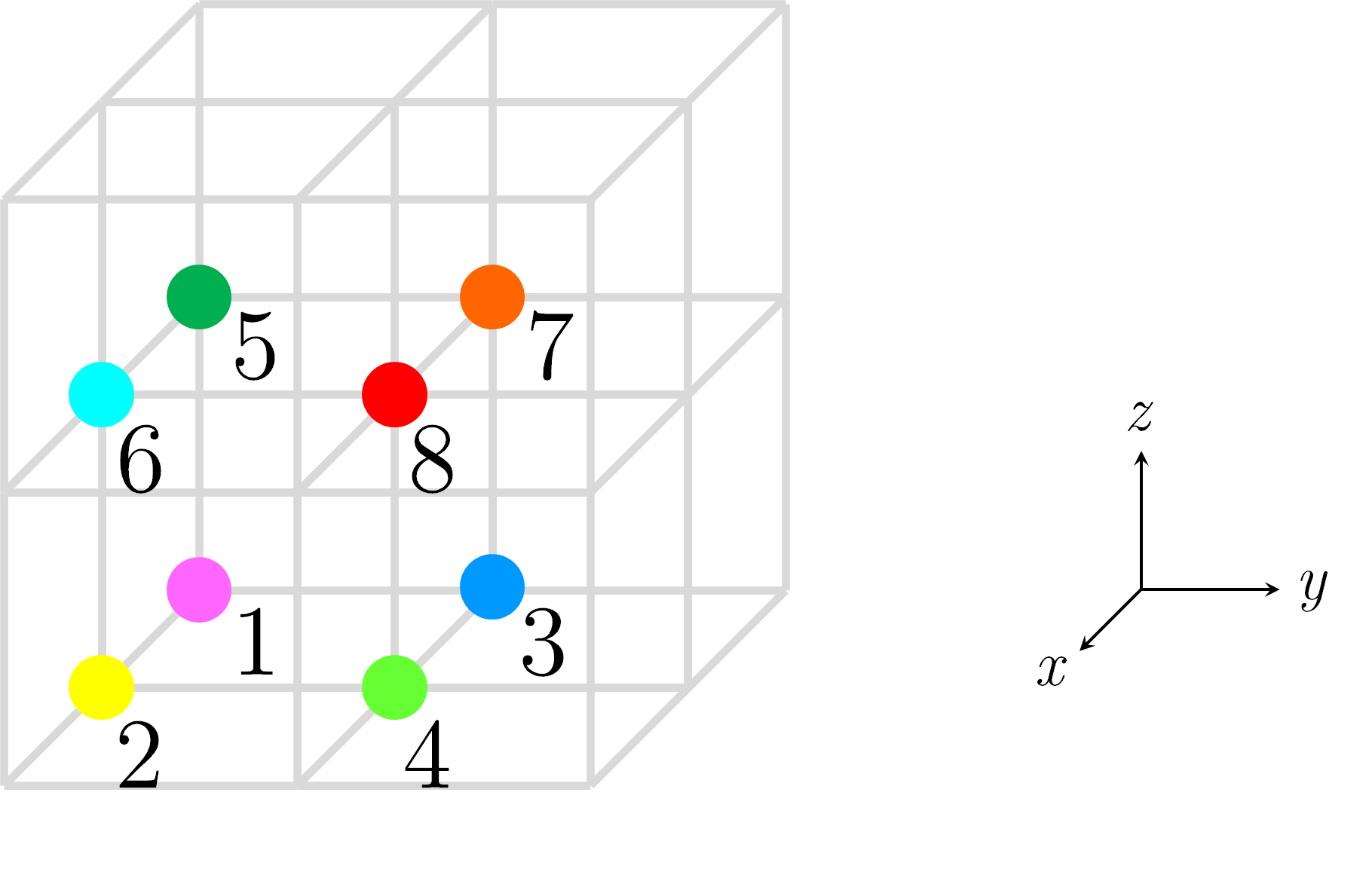}
    \caption{A $2\times2\times2$ cell of the Chamon model, regarded as a unit cell in the transformation between $H_C$ and $H_G$. There are eight qubits in the unit cell, each represented by a dot of a distinct color.}
    \label{fig:ChamonUnitCell}
\end{figure}

On the other hand, the gauged model $H_G$ has only six qubits per unit cell (one per edge of the tetrahedral-octahedral honeycomb). To match the degrees of freedom, we add two ancillary qubits per unit cell to the Hilbert space of $H_G$, forming a Hilbert space $\mathcal{H}_G$ which has eight qubits per unit cell and can thus be identified with $\mathcal{H}_C$. Each of the eight qubits is likewise labelled with a double subscript $v,i$ with $i=1,\ldots,8$. Qubits 1 through 6 are those associated with the edges emanating from $v$ in the $x$, $y$, $z$, $(0,1,-1)$, $(-1,0,1)$, and $(1,-1,0)$ directions, respectively, and 7 and 8 are the two ancillary qubits. We also add two additional terms $\overline{Z}_{v,7}\equiv Z_{v,7}$ and $\overline{Z}_{v,8}\equiv Z_{v,8}$ for each vertex $v$ to $H_G$, defining an augmented Hamiltonian $H'_G$.

To facilitate the transformation, in Fig.~\ref{fig:big} we define operators $\widehat{X}_{v,i}$ and $\widehat{Z}_{v,i}$ on $\mathcal{H}_C$ that obey relations identical to $\overline{X}_{v,i}$ and $\overline{Z}_{v,i}$ for $i=1,\ldots,8$:
\begin{equation}
\begin{split}
    [[\widehat{X}_{v,i},\widehat{X}_{v',j}]]&=[[\overline{X}_{v,i},\overline{X}_{v',j}]],\\
    \widehat{X}_{v,i}^2=\widehat{Z}_{v,i}^2=1,\qquad
    [[\widehat{X}_{v,i},\widehat{Z}_{v',j}]]&=[[\overline{X}_{v,i},\overline{Z}_{v',j}]],\\
    [[\widehat{Z}_{v,i},\widehat{Z}_{v',j}]]&=[[\overline{Z}_{v,i},\overline{Z}_{v',j}]].
\end{split}
\end{equation}
where $[[A,B]]\equiv A^{-1}B^{-1}AB$. (Each of these group commutators is a $\pm1$ phase). Due to these relations, and the fact that $\overline{Z}_{v,i}$ and $\overline{X}_{v,i}$ generate the operator algebra of $\mathcal{H}_G$, it follows that there exists an operator algebra automorphism $V$ mapping
\begin{equation}
     \overline{X}_{v,i}\to \widehat{X}_{v,i},\qquad \overline{Z}_{v,i}\to \widehat{Z}_{v,i}.
\end{equation}
Moreover, as shown in Fig.~\ref{fig:big}, $V$ maps the terms of $H_G'$ to a set of stabilizers
\begin{equation}
    \left\{\widehat{A}_t,\widehat{A}_{t'},\widehat{A}_o,\widehat{B}_{v,x},\widehat{B}_{v,y},\widehat{B}_{v,z},\widehat{Z}_{v,7},\widehat{Z}_{v,8}\right\}
\end{equation}
that generates the stabilizer group of $H_C$. Therefore,
\begin{equation}
    VH_G'V^\dagger\sim H_C
\end{equation}
where $\sim$ denotes equality of ground spaces. In the supplementary Mathematica file we demonstrate that $V$ is in fact a finite-depth Clifford circuit.
Thus, we have arrived at the first main result of the paper: the \Chamon model $H_C$ is generalized local unitary equivalent to the gauged Hamiltonian $H_G$. Appendix~\ref{app:Polynomial} provides an alternative description of this transformation in terms of the polynomial description of translation invariant Pauli stabilizer codes.

To better understand this equivalence, we consider how the transformation acts on the fractional excitation superselection sectors. First, we note that the wireframe operators of $H_G$ are mapped by $V$ into wireframe operators (with even-length edges) of the Chamon model $H_C$. Therefore, the lineons of $H_G$ become the lineons of $H_C$ under the transformation. This is consistent with the fact that both models exhibit a mutual $\pi$ braiding statistic between intersecting lineons sharing an $x$, $y$, $z$, or $w$ plane. Second, we note that the loop operators for fracton dipoles of $H_G$ are transformed into loop operators for the elementary planons of the Chamon model lying in even dual lattice planes. In other words, adjacent fracton dipoles are mapped into pairs of elementary planons of $H_C$ separated by two lattice spacings. This is consistent with the fact that the fracton dipoles of $H_G$ have fermionic exchange statistics but trivial mutual braiding statistics, as the elementary planons in the Chamon model are fermions that braid non-trivially with their nearest neighbors only.

\subsection{Weak SSPT}

In this section, based on the excitation content of $H_G$, we argue that the matter Hamiltonian $H_M$ represents a weak subsystem symmetry-protected topological (SSPT) state. A weak SSPT is defined as one that can be obtained by stacking 2D SPTs onto a trivial state in such a way that all planar symmetries are preserved~\cite{subsystemphaserel,Devakul_2020}. In the presence of fermionic degrees of freedom, this definition can be extended to allow for stacking of non-invertible 2D topological states. In particular, we consider starting with a completely trivial state (Ising paramagnet plus atomic insulator) on the matter Hilbert space of $H_M$. We then stack alternating layers of invertible topological orders corresponding to the $\nu=4$ and $\nu=-4$ states of the Kitaev 16-fold way~\cite{kitaev2006anyons} onto each plane of the tetrahedral-octahedral honeycomb. Finally, each of the $S_P$ symmetry generators is modified such that it is the product of the original $S_P$ with the total fermion parities of the two Kitaev states adjacent to $P$. It is easy to see that this modification preserves all the relations of the symmetry group. We conjecture that this state belongs to the same universality class as the model $H_M$.

To see why this is reasonable, it is helpful to consider the same construction on the gauged level, which should yield a model gLU-equivalent to the Chamon model. In the gauged system, the stacking of Kitaev states is equivalent to stacking alternating layers of fermion parity-gauged $\nu=4$ and $\nu=-4$ states, \textit{i.e.} semion-fermion and anti-semion-fermion topological orders, onto an untwisted fermionic gauge theory (equivalent to the model described by polynomial matrix $T\Sigma$ of Appendix~\ref{app:4XC}). After stacking, bound states of the emergent fermion and the fracton dipole living in the same plane are condensed, confining all of the original lineons in the model but leaving deconfined bound states formed out of a lineon fused with the a semion (or anti-semion) in each of the two parallel planes. This step is equivalent to modifying the symmetry generators $S_P$ on the ungauged level. It is clear that this procedure results in the correct braiding statistics of gauge flux planons, \textit{i.e.} a mutual semionic statistic between adjacent lineon dipoles. Each of these bound-state lineons can be mapped to a (possibly dyonic) lineon of the Chamon model, therefore the condensed model has the same fractional excitation content as the Chamon model.

\section{Foliated fracton order}
\label{sec:FFO}

A model is said to have \textit{foliated fracton order} (FFO) if its system size can be systematically reduced by disentangling, or \textit{exfoliating}, layers of 2D topological orders from the bulk system via generalized local unitary (gLU) transformation~\cite{shirley2017fracton}. If there are $n$ different orientations of such 2D states, the model is said to have an $n$-foliation structure. The first known example of FFO was the X-cube model, which has a 3-foliation structure, followed by a handful of other examples including 1-, 2-, and 3-foliated models~\cite{shirley2018CB,Wang2019,shirley2018Fractional,Shirley2019}.

In this section we demonstrate that the \Chamon model hosts 4-foliated fracton order, with foliation layers normal to the $x$, $y$, $z$, and $w=(1,1,1)$ directions. In particular, we show that the system size can be decreased by a constant factor $m$ by exfoliating stacks of 2D toric codes~\cite{qdouble} in four directions from the bulk system, where $m$ is any odd integer. This result is consistent with previous studies on entanglement signatures~\cite{Dua_Classification_2019} and compactification~\cite{Dua2019_compactify} of the model.

$H_C$ is defined on a cubic lattice, which we will take to have integer coordinates in this section and refer to as $\Lambda$. The combination of Hamiltonian and its underlying lattice is denoted $H_C(\Lambda)$. We also define coarse-grained cubic lattices $m\Lambda$ whose lattice constants are the integer $m$. For a given odd $m$, we posit the existence of a Clifford circuit $U$ satisfying
\begin{equation}
    \begin{split}
        UH_C(\Lambda)U^\dagger\sim H_C(m\Lambda)+H_{2D}(m\Lambda)
    \end{split}
    \label{eq:FFO}
\end{equation}
where $\sim$ denotes equality of ground spaces, and the Hamiltonian $H_{2D}$ describes four stacks of decoupled 2D toric codes normal to the $x$, $y$, $z$, and $w$ directions respectively, each with $\frac{m-1}{2}$ toric codes per lattice spacing. We construct such a circuit explicitly in the supplementary Mathematica file in the $m=3,5$ cases. In the case of general $m$, we show in Appendix~\ref{app:proof} the unitary $U$ exists, although we do not explicitly equate the model $H_{2D}(m\Lambda)$ to stacks of toric codes. In the following discussion, we explain the Chamon model's foliation structure (\ref{eq:FFO}) at the level of its fractional excitations.

In general, gapped long-range entangled phases are characterized by the structure of fractional excitations above the ground state. In FFOs, exfoliation of a set of 2D topological states corresponds to a factorization of the fusion group $A$ of quasiparticle superselection sectors into two subgroups $A'\boxtimes A_{2D}$. Here, we use $\boxtimes$ to denote a product of fusion groups such that there are no nontrivial mutual statistics between the two factors. $A'$ is the fusion group of the coarse-grained fracton order, and $A_{2D}$ is the fusion group of planons in the exfoliated topological layers.

\begin{figure*}[]
    \centering
    \includegraphics[width=.7\textwidth]{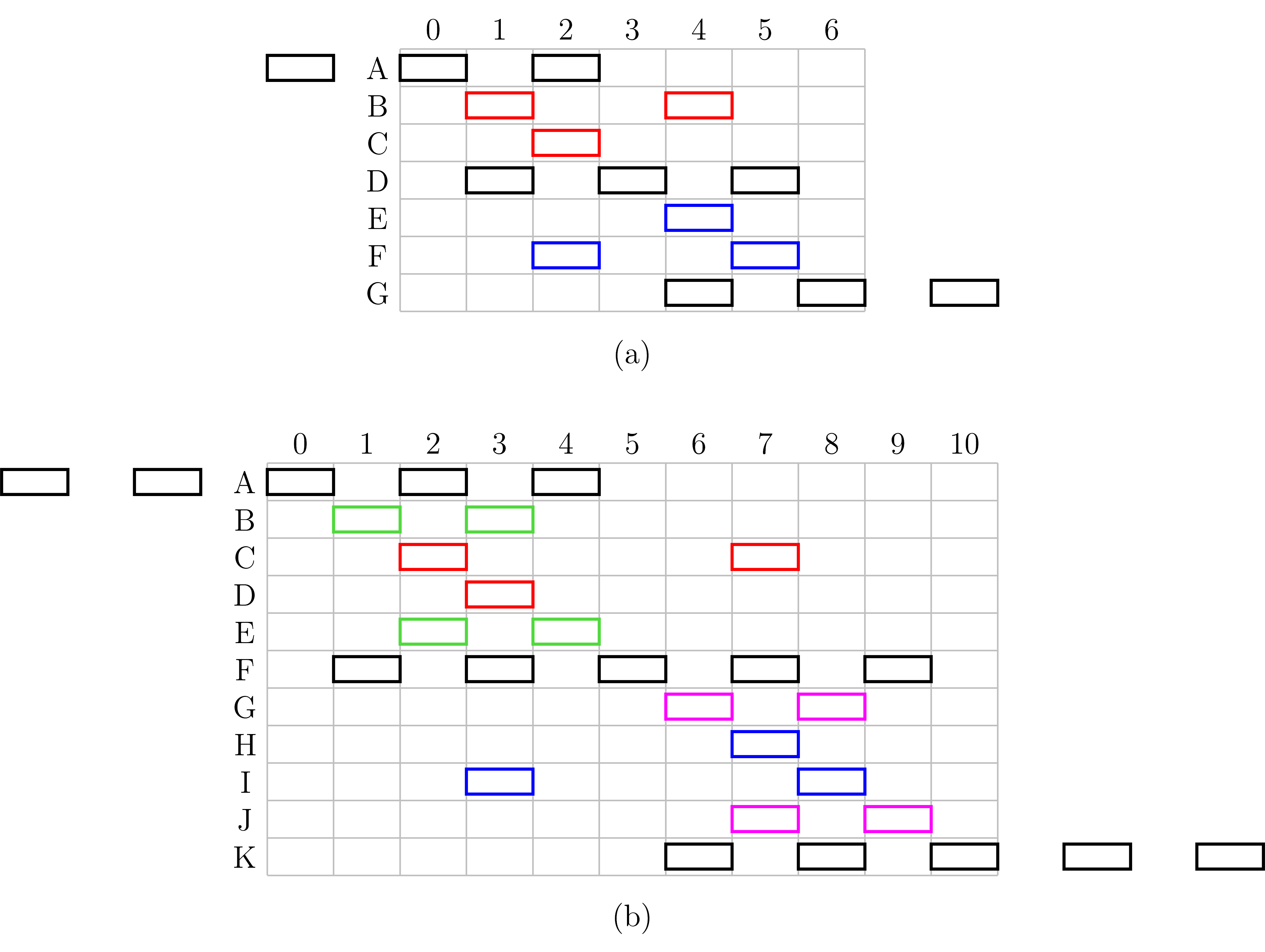}
    \caption{
      Planon diagrams depicting coarse-grained bases of $\mu$-normal planons in a given direction $\mu=x,y,z,w$ for the (a) $m=3$ and (b) $m=5$ cases. Each basis is translation-invariant with respect to the enlarged unit cell of $2m\Lambda$. Each vertical line is commensurate with a $\mu$-normal lattice plane, hence the numbers 0 to $2m$ represent dual lattice coordinates. A box lying in column $k$ represents a planon living in that dual lattice plane. On the other hand, each horizontal row represents a single generator of our chosen basis, equal to the fusion product of all elementary planons in the row. Since each unit cell contains $2m$ basis planons, A and G (K) belong to different unit cells for the $m=3$ ($m=5$) case. The planon bases are partitioned into $m$ subsets of two generators per unit cell, such that they have pairwise trivial mutual braiding statistics. (Recall that adjacent planons of the Chamon model have a mutual $\pi$ braiding statistic). For $m=3$, the subsets are colored black (ADG), red (BC), and blue (EF), whereas for $m=5$ they are colored  black (AFK), green (BE), red (CD), purple (GJ), and blue (HI). The black planons are excitations of the coarse-grained Chamon model $H_C(m\Lambda)$, as they are fermions (being composed of an odd number of fermions with trivial mutual statistics) with a mutual $\pi$ braiding statistic between adjacent pairs. On the other hand, each of the remaining $m-1$ pairs of planons generates a decoupled layer of 2D toric code. These diagrams verify the relation (\ref{eq:decomp}).}
    \label{fig:RG}
\end{figure*}

In the case of the \Chamon model, we find that the fusion group $A_C(\Lambda)$ on lattice $\Lambda$ obeys the following property:
\begin{equation}
    A_C(\Lambda)\cong A_C(m\Lambda)\boxtimes A_{2D}(m\Lambda)
    \label{eq:SDecomp}
\end{equation}
where
\begin{equation}
    A_{2D}=A^x_{2D}\boxtimes A^y_{2D}\boxtimes A^z_{2D}\boxtimes A^w_{2D}
\end{equation}
and $A^x_{2D}$, $A^y_{2D}$, $A^z_{2D}$, and $A^w_{2D}$ are the fusion groups of stacks of 2D toric codes in the $x$, $y$, $z$, and $w$ directions respectively, each with $\frac{m-1}{2}$ toric codes per lattice spacing. Here $\cong$ denotes a locality-preserving isomorphism.

To see this, note that by the transformation of the previous section, the fusion rules of $H_C(\Lambda)$ are identical to those of the 4-foliated X-cube model $H_{4XC}(2\Lambda)$ discussed in Appendix~\ref{app:4XC} (since $H_C$ is gLU equivalent to $H_G$ whose fusion rules are the same as $H_{4XC}$). The fusion group of $H_{4XC}$ is known to have the form $A_{4XC}=Q_{4XC}\times P_{4XC}$ where $P_{4XC}$ is the subgroup consisting of all planon excitations~\cite{shirley2018Fractional}, and $Q_{4XC}$ is a (non-unique) finite subgroup generated by one fracton and three lineons. As an aside, this observation forms the basis of the notion of \textit{quotient superselection sectors} (QSS), which are defined as equivalence classes of superselection sectors modulo planons~\cite{shirley2018Fractional}. According to this definition, the group of QSS of $H_{4XC}$ (and hence of $H_C$) is $A_{4XC}/P_{4XC}\cong Q_{4XC}$.

Hence, we have that $A_C=Q_C\times P_C$ where $Q_C$ is an order 16 subgroup and $P_C=P^x_C\boxtimes P^y_C\boxtimes P^z_C\boxtimes P^w_C$ is the subgroup of all planons. The decomposition of (\ref{eq:SDecomp}) is implied by the following decomposition of $P$:
\begin{equation}
    P_C(\Lambda)\cong P_C(m\Lambda)\boxtimes A_{2D}(m\Lambda),
    \label{eq:Pdec}
\end{equation}
since $Q_C$ can always be chosen such that $Q_C$ and $A_{2D}(m\Lambda)$ have no nontrivial mutual statistics, i.e.
\begin{equation}
    A_C(\Lambda)\cong \left[Q_C\times P_C(m\Lambda)\right]\boxtimes A_{2D}(m\Lambda).
\end{equation}
The equivalence (\ref{eq:Pdec}) can in turn be factored by direction:
\begin{equation}
    P_C^\mu(\Lambda)\cong P_C^\mu(m\Lambda)\boxtimes A_{2D}^\mu(m\Lambda).
\end{equation}
Thus, we can focus on the group of planons in a single direction, $P^\mu_C(\Lambda)$. Recall from Sec.~\ref{sec:Chamon} that for a given direction, there is one independent planon per lattice spacing whose loop operator is given by the product of $O_c$ terms in a particular dual lattice plane. The total group is generated by the set of all such elementary planons. Each elementary planon has fermionic exchange statistics. Moreover, neighboring planons have mutual semionic braiding statistics.

To demonstrate (\ref{eq:Pdec}), we need to find an alternative set of generating planons that splits into two parts: one that generates $P_C^\mu(m\Lambda)$ and one that generates $A_{2D}^\mu(m\Lambda)$. Actually, we will show the following equivalent relation:\footnote{The additional coarse-graining by a factor of two is necessary to pair up 3-fermion states so they can be transformed into pairs of toric codes.}
\begin{equation}
    P_C^\mu(\Lambda)\cong P_C^\mu(m\Lambda)\boxtimes A_{2D}^\mu(2m\Lambda)\boxtimes A_{2D}^\mu(2m\Lambda).
    \label{eq:decomp}
\end{equation}
Factorization of this form for $m=3$ and $m=5$ are depicted in the planon diagrams of Fig. \ref{fig:RG}, demonstrating that the fractional excitation structure of $H_C$ indeed exhibits the decomposition of (\ref{eq:SDecomp}). It is straightforward to generalize these diagrams for larger $m$. Thus, we conclude that the \Chamon model exhibits a 4-foliation structure of 2D toric code layers in the $x$, $y$, $z$, and $w$ directions.

\section{Discussion}
\label{sec:Discussion}
In this work, we have carried out a comprehensive investigation of the Chamon model, which is historically significant as the first fracton model to appear in the literature. Specifically, we have demonstrated two results: first, its characterization as a twisted 4-foliated gauge theory with emergent fermionic charge. Second, we have found that it has a 4-foliation structure composed of 2D toric code layers. The foliation structure is consistent with a conjecture of Ref.~\cite{Dua_RG_2019}, which outlines conditions under which a copy of 2D toric code can be extracted from a 3D stabilizer code model under a local unitary. The emergent gauge theory structure found in this paper, has been used by two of the authors to write a topological defect network for the \Chamon model~\cite{song2021topological}.

The transformation between the Chamon model and the 4-foliated X-cube variant $H_G$ is reminiscent of previous findings about the checkerboard model~\cite{shirley2018CB} and the Majorana checkerboard model~\cite{vijay_CB_2015}, which were respectively shown to be equivalent to two copies of the (3-foliated) X-cube model, and to the semionic X-cube model~\cite{Wang2019} (plus transparent fermions), each of which has a clear gauge theory description. It is similarly reminiscent of the equivalence between the Wen plaquette model~\cite{Wen2003} and the 2D toric code~\cite{kitaev2006anyons}. These transformations all have in common that the original model, \textit{e.g.} Chamon, has an enhanced translation symmetry compared with the transformed model, \textit{e.g.} $H_G$. Therefore, the respective gauge theory descriptions are enriched by translation symmetry via a nontrivial permutation on the fractonic superselection sectors. We leave a detailed exploration of this topic to future studies.

While it is known that CSS stabilizer codes can generically be characterized via emergent gauge theory, our results raise the question of how generally non-CSS codes in three dimensions admit such a description. It seems plausible that all stabilizer codes possess a gauge theory description and hence it could be enlightening to study more examples. For instance, one could check whether a gauge theory description, analogous to the Chamon model, is possible for the fracton models in Ref.~\cite{HHB}. Another question raised by this work is that of strong subsystem symmetry-protected topological (SPT) states in fermionic systems, whose classification is an open problem. We have argued that the Chamon model is dual to a weak subsystem SPT.

More generally, it is an open question to what extent the framework of emergent gauge theory has utility in the classification of fractonic phases of matter. To our knowledge, among the class of exactly solvable lattice models, there are no examples that are explicitly known to not admit a gauge theory description. It would be interesting to either find such an example, or demonstrate that none exist. On the other hand, there are examples of fractonic orders with excitations of infinite order which are unlikely to have any characterization in terms of finite gauge groups (although they arise naturally as infinite-component $U(1)$ Chern-Simons gauge theories~\cite{ChernSimons}).

Finally, it is worthwhile to note that the some of the fractonic excitations in the Chamon model exhibit non-bosonic self-exchange statistics~\cite{FractonSelfStatistics}. For the present analysis it has been sufficient to consider in detail the statistics of planon excitations. A systematic investigation of fracton self-statistics in $n$-foliated models is left to future work.



\section*{Acknowledgments}
W.S. is grateful to Xie Chen and Kevin Slagle for many inspiring discussions. A.D. thanks Dominic J. Williamson for related discussions. W.S. and A.D. are supported by the Simons Foundation through the collaboration on Ultra-Quantum Matter (651444, WS; 651438, AD) and by the Institute for Quantum Information and Matter, an NSF Physics Frontiers Center (PHY-1733907). W.S. also received support from the National Science Foundation under the award number DMR-1654340.

\bibliography{bib}
\bibliographystyle{apsrev4-2}
\appendix

\onecolumngrid




\section{Relation between $H_G$ and the 4-foliated X-cube model}
\label{app:4XC}
In this section we introduce the 4-foliated X-cube model $H_{4XC}$ and describe its relation to $H_G$. The fusion structure of excitations of $H_{4XC}$ is identical to that of $H_G$. However, the models differ in terms of the self and mutual statistics of the excitations. In this section we will use the $\mathbb{Z}_2[x,y,z,1/x,1/y,1/z]$ Laurent polynomial ring formalism for describing translation-invariant Pauli stabilizer codes~\cite{haah2013commuting}. In this formalism, Pauli operators in a cubic lattice system with $n$ qubits per site are represented by length $2n$ column vectors whose entries are elements of $\mathbb{Z}_2[x,y,z,1/x,1/y,1/z]$. The first $n$ entries represent the Pauli $X$ components, and the last $n$ entries the Pauli $Z$ components.

The Hilbert space of $H_{4XC}$ is the same as that of $H_G$. It is composed of one qubit on each edge of the tetrahedral-octahedral honeycomb. The Hamiltonian has the form
\begin{equation}
    H_{4XC}=-\sum_cA_c-\sum_{v,\mu}B_{v,\mu}
\end{equation}
where $c$ runs over all 3-cells of the honeycomb, $v$ all vertices, $\mu=x,y,z,w$, and
\begin{equation}
    A_c=\prod_{e\in c}X_e,\qquad B_{v,\mu}=\prod_{v\ni e\perp\mu}Z_e.
\end{equation}
This model is described by the polynomial matrix
\begin{equation}
    \Sigma=\begin{pmatrix}
    A&0\\0&B
    \end{pmatrix}
\end{equation}
where
\begin{equation}
    A=\begin{pmatrix}
        1 & yz & y+z\\
        1 & zx & z+x \\
        1 & xy & x+y \\
        1 & x & x+1 \\
        1 & y & y+1 \\
        1 & z & z+1
    \end{pmatrix},\qquad
    B=\begin{pmatrix}
        0 & 1+\frac{1}{x} & 1+\frac{1}{x} \\
        1+\frac{1}{y} & 0 & 1+\frac{1}{y} \\
        1+\frac{1}{z} & 1+\frac{1}{z} & 0 \\
        \frac{1}{y}+\frac{1}{z} & 0 & 0 \\
        0 & \frac{1}{x}+\frac{1}{z} & 0 \\
        0 & 0 & \frac{1}{x}+\frac{1}{y}
    \end{pmatrix}.
\end{equation}
The columns of $A$ represents the 3-cell terms $A_t$, $A_{t'}$ and $A_o$, whereas the columns of $A$ represent the vertex terms $B_{v,x}$, $B_{v,y}$, and $B_{v,z}$, which together generate $B_{v,w}$. Note that $\Sigma^\dagger\Omega_6\Sigma=0$ where $^\dagger$ represents transposition combined with spatial inversion, $\Omega_{k}=\begin{pmatrix}
0 & I_k \\
I_k & 0 
\end{pmatrix}$ is the $2k\times 2k$ symplectic form, and $I_k$ the $k\times k$ identity matrix. Thus, the terms of $H_{4XC}$ are mutually commuting.

\begin{figure*}[tbp]
    \centering
    \includegraphics[width=.9\textwidth]{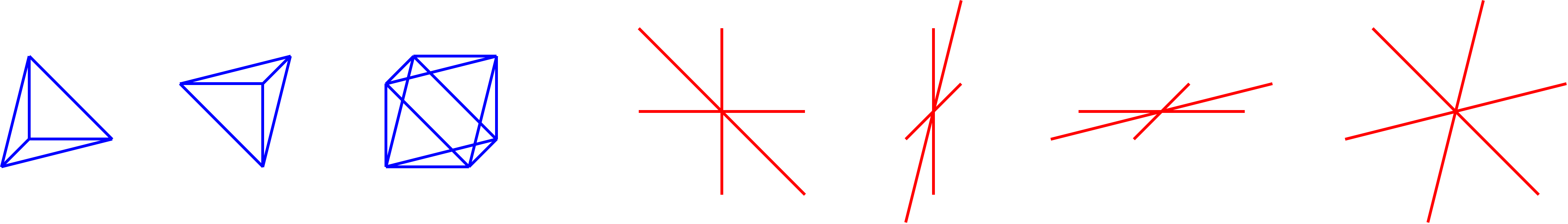}
    \caption{The terms $A_t$, $A_{t'}$, $A_o$, $B_{v,x}$, $B_{v,y}$, $B_{v,z}$ and $B_{v,w}$ of $H_{4XC}$, where $t$ and $t'$ are oppositely oriented tetrahedral cells and $o$ an octahedral cell. Here blue represents Pauli $X$ and red Pauli $Z$. Each term is a tensor product of the depicted Pauli operators.}
    \label{fig:4XC}
\end{figure*}

$H_G$ can be obtained from $H_{4XC}$ via a pair of locality-preserving, invertible but non-isomorphic transformations of the Pauli group $\mathcal{P}$:
\begin{equation}
    W\colon \mathcal{P}\to\mathcal{P},\qquad
    T\colon \mathcal{P}\to\mathcal{P}.
    \label{eq:algebra}
\end{equation}
In the polynomial formalism, these transformations correspond to multiplication by invertible but non-symplectic matrices:
\begin{equation}
    W=\begin{pmatrix}
    I_6&0\\ \widetilde{W}&I_6
    \end{pmatrix},\qquad
    T=\begin{pmatrix}
    I_6&\widetilde{T}\\0&I_6
    \end{pmatrix}
\end{equation}
where
\begin{equation}
    \widetilde{W}=\begin{pmatrix}
    1 & 1 & 0 & z & z & 0 \\
    0 & 0 & 0 & 0 & 0 & 0 \\
    1 & 1 & 0 & z & z & 0 \\
    \frac{1}{y} & \frac{1}{y} & 0 & \frac{z}{y} &     \frac{z}{y} & 0 \\
    0 & 0 & 0 & 0 & 0 & 0 \\
    \frac{1}{y} & \frac{1}{y} & 0 & \frac{z}{y} &     \frac{z}{y} & 0 
    \end{pmatrix},\qquad
    \widetilde{T}=\begin{pmatrix}
    0 & 1+\frac{y}{x} & 0 & 0 & z+1 & 0 \\
    0 & 0 & 0 & 0 & 0 & 0 \\
    1+\frac{x}{z} & 1+\frac{y}{z} & 0 & 0 & x+1 & 0 \\
    1+\frac{x}{yz} & 1+\frac{1}{z} & 1+\frac{1}{y} & 0 & 1+\frac{x}{y} & 0 \\
    0 & 1+\frac{y}{xz} & 0 & 0 & 0 & 0 \\
    1+\frac{1}{y} & 1+\frac{1}{x} & 1+\frac{z}{xy} & 1+\frac{z}{x} & 1+\frac{z}{y} & 0 
    \end{pmatrix}.
\end{equation}
Note that $T^2=W^2=1$. It holds that
\begin{equation}
    \overline{A}_c=T(W(A_c)),\qquad\overline{B}_{v,\mu}=T(W(B_{v,\mu})).
\end{equation}
Therefore, the Hamiltonian $H_G$ is represented by the polynomial matrix $\overline{\Sigma}=TW\Sigma$. Since $\overline{\Sigma}^\dagger\Omega_6\overline{\Sigma}=0$, the terms of $H_G$ mutually commute hence defining a stabilizer code. Note that
\begin{equation}
    TW=\begin{pmatrix}
    I_6+\tilde{T}\tilde{W}&\tilde{T}\\\tilde{W}&I_6
    \end{pmatrix}.
\end{equation}

The $T$ and $W$ transformations can also be used to define two other non-CSS stabilizer code Hamiltonians, represented by the polynomial matrices $W\Sigma$ and $T\Sigma$ satisfying $\Sigma^\dagger W^\dagger\Omega_6W\Sigma=0$ and $\Sigma^\dagger T^\dagger\Omega_6T\Sigma=0$. The Hamiltonians represented by $\Sigma$, $\overline{\Sigma}$, $W\Sigma$, and $T\Sigma$ can each be obtained via a gauging procedure of four stacks of planar $\mathbb{Z}_2$ symmetries. The procedure was described explicitly for $H_G$, represented by $\overline{\Sigma}$, in Sec.~\ref{sec:Gauging}. On the other hand, $\Sigma$, $W\Sigma$, and $T\Sigma$ can be obtained by gauging the following matter Hamiltonians, respectively:
\begin{align}
    H_M^{(1)}&=-\sum_tX_t-\sum_oX_o,\\
    H_M^{(2)}&=-\sum_tX_t-\sum_o\mathcal{X}_o,\\
    H_M^{(3)}&=-\sum_ti\gamma_t\gamma'_t-\sum_oX_o.
\end{align}
The Hilbert space of $H_M^{(3)}$ is the same as that of $H_M$, whereas those of $H_M^{(1)}$ and $H_M^{(2)}$ differ in that the fermionic orbital on each tetrahedral 3-cell is replaced by a qubit. The symmetries of $H_M^{(3)}$ are the same as those of $H_M$, whereas for $H_M^{(1)}$ and $H_M^{(2)}$ they are simply a product of Pauli $X$ operators over all 3-cells in a given dual lattice plane.

Therefore, each of the models $\Sigma$, $\overline{\Sigma}$, $W\Sigma$, and $T\Sigma$ represents a distinct kind of fractonic gauge theory. $\Sigma$ is coupled to a trivial bosonic paramagnet, $T\Sigma$ to a trivial atomic insulator/paramagnet state, $W\Sigma$ to a bosonic SSPT state, and $\overline{\Sigma}$ to a fermionic SSPT. The fusion rules of all four models are identical; moreover the generalized Aharanov-Bohm statistics between gauge charge and flux sectors have identical form. However, the models differ in terms of the statistics \textit{within} the charge and flux sectors. Acting on $\Sigma,$ the $W$ matrix represents a \textit{twist} of the gauge flux statistics, whereas the $T$ matrix represents a \textit{transmutation} of the gauge charge statistics. 
This can be seen from the equations
\begin{equation}
    W^\dagger\Omega_{6} W=\begin{pmatrix}
        \widetilde{W}+\widetilde{W}^\dagger&I_6\\I_6&0
    \end{pmatrix},\qquad
    T^\dagger\Omega_{6} T=\begin{pmatrix}
        0&I_6\\I_6&\widetilde{T}+\widetilde{T}^\dagger
    \end{pmatrix},\qquad
    W^\dagger T^\dagger\Omega_{6} TW=\begin{pmatrix}
        \widetilde{W}+\widetilde{W}^\dagger&I_6\\I_6&\widetilde{T}+\widetilde{T}^\dagger
    \end{pmatrix}.
    \label{eqn:TWprops}
\end{equation}
The off-diagonal components represent the Aharanov-Bohm interactions whereas the diagonal components represent the statistics within the charge and flux sectors. Therefore, $\Sigma$ and $W\Sigma$ have purely bosonic gauge charge statistics, whereas the tetrahedral fractonic charges of $\overline{\Sigma}$ and $T\Sigma$ are fermionic. On the other hand, $\Sigma$ and $W\Sigma$ have purely bosonic gauge flux lineons, whereas intersecting lineons of $\overline{\Sigma}$ and $W\Sigma$ have a mutual semionic braiding statistic.

\section{Polynomial representation of the transformation from $H_G$ to $H_C$}
\label{app:Polynomial}

In this appendix we express the transformation from the gauge theory Hamiltonian $H_G$ to the Chamon model $H_C$ in terms of the Laurent polynomial formalism. Regarding a $2\times2\times2$ cell as the unit cell with qubits labelled as in Fig.~\ref{fig:ChamonUnitCell}, $H_C$ is represented by the $16\times 8$ stabilizer map
\begin{equation}
    \widehat{\Sigma}=\begin{pmatrix}
    0 & x & y & 0 & 0 & xz & yz & 0 \\
    1 & 0 & 0 & y & z & 0 & 0 & yz \\
    1 & 0 & 0 & x & z & 0 & 0 & xz \\
    0 & 1 & 1 & 0 & 0 & z & z & 0 \\
    0 & x & y & 0 & 0 & x & y & 0 \\
    1 & 0 & 0 & y & 1 & 0 & 0 & y \\
    1 & 0 & 0 & x & 1 & 0 & 0 & x \\
    0 & 1 & 1 & 0 & 0 & 1 & 1 & 0 \\
    0 & 0 & y & xy & z & xz & 0 & 0 \\
    0 & 0 & y & y & z & z & 0 & 0 \\
    1 & x & 0 & 0 & 0 & 0 & z & xz \\
    1 & 1 & 0 & 0 & 0 & 0 & z & z \\
    1 & x & 0 & 0 & 0 & 0 & y & xy \\
    1 & 1 & 0 & 0 & 0 & 0 & y & y \\
    0 & 0 & 1 & x & 1 & x & 0 & 0 \\
    0 & 0 & 1 & 1 & 1 & 1 & 0 & 0 
    \end{pmatrix}.
\end{equation}
We define a matrix $C$ whose first (last) 8 columns represent the operators $\widehat{X}_{v,i}$ ($\widehat{Z}_{v,i}$) for $i=1,\ldots,8$:
\begin{equation}
    C=\begin{pmatrix}
    1 & 1 & 0 & z & z & 0 & 0 & 1 & 0 & 0 & 0 & 0 & 0 & 0 & 1 & 0 \\
    1 & 0 & 0 & 0 & 0 & 0 & 0 & 1 & 0 & 1+\frac{y}{x} & \frac{z}{x} & \frac{yz}{x} & z & y & 0 & 0 \\
    0 & 1 & 0 & 0 & 0 & 0 & 0 & 0 & 0 & 0 & \frac{z}{y} & z & \frac{xz}{y} & 0 & 0 & 1+\frac{x}{y} \\
    0 & 0 & 0 & 0 & 0 & 0 & 0 & \frac{1}{y} & 0 & 0 & 0 & 0 & 0 & 1 & \frac{1}{y} & 0 \\
    0 & 0 & 0 & 0 & 0 & 0 & 0 & 1 & 0 & 0 & 1 & 0 & 0 & 0 & 1 & 0 \\
    0 & 0 & 0 & 0 & 1 & 0 & 0 & 0 & 1 & 1 & 1 & 0 & 0 & 0 & 0 & 0 \\
    0 & 0 & 0 & 1 & 0 & 0 & 0 & \frac{1}{y} & 1 & 1 & 1 & 0 & 1+\frac{x}{y} & 1 & 0 & \frac{1}{z}+\frac{x}{yz} \\
    0 & 0 & 0 & 0 & 0 & 0 & 0 & 0 & 0 & 0 & \frac{1}{y} & 0 & 0 & 1 & \frac{1}{y} & 0 \\
    0 & 1 & 1 & y+z & z & y & 0 & 1 & 0 & 0 & 0 & 0 & 0 & 0 & 1 & 1 \\
    0 & 0 & 0 & 0 & 0 & 0 & 0 & 1 & 1 & 0 & 0 & 0 & 0 & 0 & 1 & 1 \\
    0 & 1 & 0 & 0 & 0 & 0 & 0 & 0 & 0 & 1 & \frac{z}{y} & z & \frac{xz}{y} & 0 & 0 & \frac{x}{y} \\
    0 & 0 & 0 & 0 & 0 & 1 & \frac{1}{y} & \frac{1}{y} & 1+\frac{1}{y} & 1+\frac{1}{x} & 1+\frac{z}{xy} & 1+\frac{z}{x} & 1+\frac{z}{y} & 1 & 0 & \frac{1}{y} \\
    0 & 0 & 1 & 0 & 0 & 0 & 1 & 1 & 0 & 1 & 1 & 0 & x & 0 & 0 & \frac{x}{z} \\
    0 & 0 & 0 & 0 & 1 & 0 & 0 & 0 & 1 & 1 & 1 & 0 & 1 & 0 & 0 & \frac{1}{z} \\
    0 & 0 & 0 & 0 & 0 & 0 & 0 & \frac{1}{y} & 0 & 0 & 0 & 1 & 0 & 0 & \frac{1}{y} & \frac{1}{z} \\
    0 & 0 & 0 & 0 & 0 & 0 & \frac{1}{y} & 0 & \frac{1}{y} & 0 & 0 & 1 & 0 & 0 & \frac{1}{y} & \frac{1}{z} 
    \end{pmatrix}.
\end{equation}
In this section, we will redefine the matrices $W$ and $T$ from the previous appendix such that they accommodate the two ancillary qubits. In particular,
\begin{equation}
    W=\begin{pmatrix}
    I_8&0\\ \widetilde{W}\oplus I_2&I_8
    \end{pmatrix},\qquad
    T=\begin{pmatrix}
    I_8&\widetilde{T}\oplus I_2\\0&I_8
    \end{pmatrix}
\end{equation}
Then, we define a matrix $V=CWT$ satisfying $V^\dagger\Omega_8 V=\Omega_8$ and $VTW=C$. Therefore $V$ is a Clifford QCA that maps $\overline{X}_{v,i}\to\widehat{X}_{v,i}$ and $\overline{Z}_{v,i}\to\widehat{Z}_{v,i}$. Moreover,
\begin{equation}
    V\overline{\Sigma}=\widehat{\Sigma}V_2
\end{equation}
where $V_2$ is the invertible matrix
\begin{equation}
    V_2=\begin{pmatrix}
    1 & 0 & 0 & \frac{1}{y} & 0 & \frac{1}{y} & 0 & 0 \\
    0 & 0 & 1 & 0 & 0 & \frac{1}{x} & 0 & 0 \\
    0 & 0 & 1 & \frac{1}{y} & \frac{1}{y} & \frac{1}{y} & \frac{1}{y} & 0 \\
    0 & 0 & 1 & 0 & \frac{1}{xy} & \frac{1}{xy} & 0 & 0 \\
    0 & 0 & 1 & \frac{1}{z} & 0 & 0 & 0 & \frac{1}{z} \\
    0 & 0 & 1 & 0 & 0 & 0 & 0 & 0 \\
    0 & 0 & 1 & \frac{1}{yz} & \frac{1}{yz} & 0 & 0 & 0 \\
    0 & 1 & 0 & 0 & \frac{1}{yz} & 0 & 0 & \frac{1}{yz} 
\end{pmatrix}.
\end{equation}
Therefore, $V$ maps the ground space of $H_G$ to that of $H_C$. In the supplementary Mathematica file, we demonstrate that $V$ is actually a finite-depth Clifford circuit (\textit{i.e.}, it can be decomposed into a product of elementary symplectic transformations). This demonstrates that $H_G$ and $H_C$ are gLU equivalent.

\section{Entanglement renormalization of the \Chamon model}
\label{app:proof}
In this section, we study the entanglement renormalization (ER) on the \Chamon model using the polynomial formalism~\cite{haah2013commuting,haah2014bifurcation}. The stabilizer map $\sigma$ and the excitation map $\epsilon$ for the \Chamon model can be written as 
\begin{eqnarray}
  \sigma=  \left(\begin{array}{cccc}
         (1+x^{-1})(y^{-1}+z^{-1})\\
         (1+y^{-1})(x^{-1}+z^{-1})
    \end{array}
    \right)
    \label{Chamon_code}
\end{eqnarray}
and
\begin{eqnarray}
  \epsilon= \sigma^\dagger \Omega_1= \left(\begin{array}{cccc}
         (1+y)(x+z)& (1+x)(y+z)
    \end{array}
    \right) \, .
    \label{Chamon_exc_map}
\end{eqnarray}
respectively~\cite{haah2013commuting}. Our approach to doing ER involves going to a basis of stabilizer terms such that the associated basis excitations include the bosonic planon charges. Then we write the creation operators or movers of these bosonic charges and apply translation invariant gates (up to coarse-graining) to reduce them into a canonical form of unit vectors. The excitations that form the bosonic planons and the relative positions between them are shown in Fig.~\ref{fig:RG}.
Before stating an explicit ER result for the \Chamon model, we first prove that a coarse-grained copy of itself can be extracted under ER of the \Chamon model. 
In particular, we have the following theorem.

\begin{theorem}
For any odd $m$, there exists a Clifford circuit $U$ such that
\begin{equation}
    UH_C(\Lambda)U^\dagger\sim H_C(m\Lambda)+H_B(m\Lambda),
\end{equation}
for some Pauli Hamiltonian $H_B$. Here $\sim$ denotes equality of ground spaces.
\label{thm:ERG_Chamon}
\end{theorem}


\begin{proof}
We first write down two fracton creation operators,
$$s_1=x^{m-1}(1+y+...+y^{m-1})(1+z/x+...+(z/x)^{m-1})(1,0)^T$$
and
$$s_2=y^{m-1}(1+x+...+x^{m-1})(1+z/y+...+(z/y)^{m-1})(0,1)^T,$$
which create fracton excitations at the sites corresponding to the polynomials
$(1+y^m)(x^m+z^m)$ and $(1+x^m)(y^m+z^m)$ respectively. Note that $s_1$ and $s_2$ are related via permutation of $x$ and $y$. In other words, the action of the excitation map as defined in Eq.~\ref{Chamon_exc_map} on operators $s_1$ and $s_2$ is given by
\begin{align*}
 \epsilon s_1=(1+y^m)(x^m+z^m)\\
 \epsilon s_2=(1+x^m)(y^m+z^m). 
\end{align*}
Under coarse-graining of the lattice, the translation group is reduced such that the translation variables modify to $x^\prime=x^m,y^\prime=y^m,z^\prime=z^m$. On the coarse-grained lattice, the representation of the creation operators $s_1$ and $s_2$ is given by $s_1^{(m)}$ and $s_2^{(m)}$ respectively. Namely, $\phi^m_{\#}(s_1)=s_1^{(m)}$ and $\phi^m_{\#}(s_2)=s_2^{(m)}$ where $\phi^m_{\#}$ is the map that implements coarse-graining by a factor of $m$.   
We now state two lemmas about $s_1^{(m)}$ and $s_2^{(m)}$, one about the commutation relation and the other about reducing them to a canonical form via elementary symplectic transformations. The proofs are these lemmas are given after this proof.

\begin{lemma}
\label{lem:s1s2com}
For odd $m$, $s_1^{(m)\dagger}\Omega_ms_2^{(m)}=1$ where $\Omega_m=\begin{pmatrix}
0 & \mathbbm{1} \\
\mathbbm{1} & 0 
\end{pmatrix}$ is a $2m\times 2m$ symplectic form and $\mathbbm{1}$ is an $m\times m$ Identity matrix. 
\end{lemma}

\begin{lemma}
\label{lem:s1s2canonical}
For odd $m$, the creation operators $s_1^{(m)}$ and $s_2^{(m)}$ can be mapped to
\begin{align}
s_1^{(m)} &= \begin{pmatrix} 1 & 0 \cdots & 0 &| & 0 &\cdots & 0\end{pmatrix}^T, \nonumber\\
s_2^{(m)} &= \begin{pmatrix} 0 & 0\cdots &0 &| & 1 & \cdots & 0 \end{pmatrix}^T
\end{align}
via translation invariant elementary symplectic transformations. Here, as shown, $s_1^{(m)}$ and $s_2^{(m)}$, respectively, have only one nonzero entry at the 1st and ($m^3$+1)-th vector components. 
\end{lemma}


The excitation represented as a singleton element, $(1)$ before coarse-graining, is represented by the unit vector $e_1=(1,0,0,...,0)^T$ with $m$ entries after coarse graining. Considering the action of $\epsilon$ on the creation operators $s_1^{(m)}$, $s_2^{(m)}$ yields $\epsilon s_1^{(m)}=(1+y')(x'+z')e_1$ and $\epsilon s_2^{(m)}=(1+x')(y'+z')e_1$, the excitation map becomes
\begin{align}
\epsilon = \begin{pmatrix}[ccccc|ccccc]
(1+y)(x+z) & \star & \star & \cdots & \star & (1+x)(y+z) &  \star & \star & \cdots & \star \\
0     &\star & \star & \cdots & \star  & 0   &  \star & \star & \cdots & \star \\
\vdots&\vdots& \vdots & \vdots & \vdots & \vdots \\
0    & \star & \star & \cdots & \star  & 0   &  \star & \star & \cdots & \star
\end{pmatrix} \label{eq:eps1}
\end{align}
where we suppressed the $\prime$ in the coarse-grained translation variables and where $\star$ indicates unknown entries.
Since $$\begin{pmatrix}
(1+x^{-1})(y^{-1}+z^{-1}) &0& 0 & \cdots & 0&(1+y^{-1})(x^{-1}+z^{-1}) &0& 0 & \cdots & 0
\end{pmatrix}^T\in \ker \epsilon,$$ the topological order condition $\ker \epsilon = \im \sigma=\im \Omega_q \epsilon^\dagger$ implies that
the rows of $\epsilon$ must generate
$$\begin{pmatrix}
(1+y)(x+z) &0& 0 & \cdots & 0&(1+x)(y+z) &0& 0 & \cdots & 0
\end{pmatrix}^T.$$
This implies that we can insert this as a row in the excitation map as follows, 
\begin{align}
\epsilon = \begin{pmatrix}[ccccc|ccccc]
(1+y)(x+z) & 0 & 0 & \cdots & 0 & (1+x)(y+z) &  0 & 0 & \cdots & 0 \\
(1+y)(x+z) & \star & \star & \cdots & \star & (1+x)(y+z) &  \star & \star & \cdots & \star \\
0     &\star & \star & \cdots & \star  & 0   &  \star & \star & \cdots & \star \\
\vdots&\vdots& \vdots & \vdots & \vdots & \vdots \\
0    & \star & \star & \cdots & \star  & 0   &  \star & \star & \cdots & \star 
\end{pmatrix}. \label{eq:eps2}
\end{align}
On applying appropriate row operations, we get
\begin{align}
\epsilon = \begin{pmatrix}[ccccc|ccccc]
(1+y)(x+z) & 0 & 0 & \cdots & 0 & (1+x)(y+z) &  0 & 0 & \cdots & 0 \\
0 & \star & \star & \cdots & \star & 0 &  \star & \star & \cdots & \star \\
0     &\star & \star & \cdots & \star  & 0   &  \star & \star & \cdots & \star \\
\vdots&\vdots& \vdots & \vdots & \vdots & \vdots \\
0    & \star & \star & \cdots & \star  & 0   &  \star & \star & \cdots & \star 
\end{pmatrix}. \label{eq:eps3}
\end{align}
Thus, we have extracted a copy of the \Chamon model.
\end{proof}

We now give proofs of the two lemmas that were used in proving Theorem \ref{thm:ERG_Chamon}.

\begin{proof}[Proof of Lemma \ref{lem:s1s2com}]

The polynomial given by
$s_1^\dagger \Omega_1 s_2$ encodes the commutation of translates of $s_1$ and $s_2$. Here,  $\Omega_m=\begin{pmatrix}
0 & \mathbbm{1} \\
\mathbbm{1} & 0 
\end{pmatrix}$ is an $2m\times 2m$ symplectic form and $\mathbbm{1}$ is an $m\times m$ Identity matrix.
Let us denote the coefficient of $g$ in the polynomial $s_1^\dagger \Omega_1 s_2$ as $(s_1^\dagger \Omega_1 s_2)_g$. We note that two Pauli operators $a$ and $b$ commute if $(a^\dagger \Omega_q b)_1=0$. 

Note that $s_1$ and $s_2$ can be expressed as follows,
$$s_1=(1+y+...+y^{m-1})(x^{m-1}+zx^{m-2}+...+z^{m-1})(1,0)^T$$
and
$$s_2=(1+x+...+x^{m-1})(y^{m-1}+zy^{m-2}+...+z^{m-1})(0,1)^T.$$ Since all powers of translation variables are less than $m$, under coarse-graining by a factor of $m$ in each direction, we are left with $2m$-dimensional vectors for $s_1$ and $s_2$ with only $1$s and $0$s. For $s_1$, the $1$s appear in the first half and for $s_2$, in the second half. Due to this form, $s_1^{(m)\dagger}\Omega_ms_2^{(m)}=(s_1^{(m)\dagger}\Omega_ms_2^{(m)})_1$ i.e. only the coefficient of 1 contributes and there are no monomials. Since the commutation relation between the operators $s_1$ and $s_2$ i.e. $ (s_1^{\dagger}\Omega_1 s_2)_1$ is not affected by coarse-graining, we get 
\beq
&&s_1^{(m)\dagger}\Omega_ms_2^{(m)}\nonumber\\
&=&(s_1^{(m)\dagger}\Omega_ms_2^{(m)})_1 \nonumber\\
&=& (s_1^{\dagger}\Omega_1 s_2)_1  \nonumber\\
&=& m \mod 2 
\eeq 
Thus, $s_1^{(m)\dagger}\Omega_ms_2^{(m)}=1$ when $m$ is odd.

\end{proof}

\begin{proof}[Proof of Lemma \ref{lem:s1s2canonical}]

For both $s_1$ and $s_2$, the degrees of translation variables $x$, $y$ and $z$ range from 0 to $m-1$.  Thus, after coarse-graining, $s_1^{(m)}$ and $s_2^{(m)}$ are both supported on at only one unit cell (at location 1). In particular, $s_1^{(m)}$ is a Laurent polynomial vector over $\FF_2[1]$,
satisfying $s_1^{(m)\dagger} \Omega_m s_1^{(m)} = 0$.
Since $\FF_2[1]$ is a principal ideal domain,
we can find an elementary symplectic transformation $E_1$
composed of CNOT gates that turns~$s_1$ into a vector with a single nonzero component, say, $g$ at the first entry. Since the only nonzero component in $\FF_2[1]$ is 1, $g=1$.

Since the transformation $E_1$ acts only at the origin, $E_1s_2^{(m)}$ still acts only at location 1 and thus is a Laurent polynomial vector over $\FF_2[1]$.  Since $s_1^{(m)\dagger}\Omega_m s_2^{(m)}$= 1, the $(m^3+1)$-th component of $E_1s_2^{(m)}$ must be 1. Since $E_1s_2^{(m)}$ can have non-zero entries i.e. $1$s only in the second half of the vector, they can all be cancelled out via CNOT gates without affecting the form of $s_1^{(m)}$. Thus, the we get the form of $s_1^{(m)}$ and $s_2^{(m)}$ as desired.
\end{proof}

\subsection{Explicit ER circuit}
In the supplementary Mathematica file, we have constructed a circuit $U$ which carries out an explicit ER of the \Chamon model given as follows: 
\begin{align}
        &UH_C(\Lambda)U^\dagger\sim H_C(3\Lambda)+H_{2D}(6\Lambda)+H_{2D}(6\Lambda)\nonumber,\\
        &H_{2D}=H^{toric}_x+H^{toric}_y+H^{toric}_z+H_w^{toric}. 
\end{align}
Here, $H^{toric}_\mu$ is a stack of 2D toric codes along the $\mu$ direction with one layer per lattice spacing. 
\end{document}